\documentclass[11pt]{article} 

\usepackage{amsmath}
\usepackage{amssymb}
\usepackage{amsthm}

\usepackage{bbm}		% To get the double struck, or blackboard bold, 1.
\usepackage{physics}	% To get good Dirac notation, among other things.

\usepackage[linesnumbered,ruled]{algorithm2e}

\usepackage{color}
\definecolor{darkblue}{rgb}{0,0,0.5}
\definecolor{darkgreen}{rgb}{0,0.5,0}

\usepackage[colorlinks=true,linkcolor=darkblue,urlcolor=darkblue,citecolor=darkgreen]{hyperref}
\usepackage[nameinlink,capitalise]{cleveref}

\usepackage{tikz}
\usetikzlibrary{calc, positioning, datavisualization}
\usetikzlibrary{datavisualization.formats.functions}

\newcommand{\N}{\mathbb{N}}
\newcommand{\R}{\mathbb{R}}
\newcommand{\C}{\mathbb{C}}

\newcommand{\mc}[1]{\mathcal{#1}}
\newcommand{\xor}{\oplus}
\newcommand{\tensor}{\otimes}
\newcommand{\Qb}{\mathcal{Q}}		% Notation for the state space of a qubit.
\DeclareMathOperator{\Bool}{Bool}	% $\Bool(n,m) = \{f : \{0,1\}^n \to \{0,1\}^m\}$.
\newcommand{\Id}{\mathbbm{1}}		% Identity operator.
\newcommand{\HS}{\mc{H}}			% Generic symbol for a Hilbert space.
\newcommand{\negl}{\mathrm{negl}}	% negligible

\newcommand{\E}{\mathop{{}\mathbb{E}}}

\newtheoremstyle{Definition_Line_Break}
	{\topsep}				% Space above theorem.
	{\topsep} 				% Space below theorem.
	{}						% Font name for body of theorem.
	{}						% Space to indent.
	{\bfseries}				% Font name for theorem head.
	{.}						% Punctuation between head and body.
	{\newline}				% Space after theorem head.
	{}						% Specify theorem head.

\theoremstyle{Definition_Line_Break}
\newtheorem{definition}{Definition}

\newtheoremstyle{Theorem_Line_Break}
	{\topsep}				% Space above theorem.
	{\topsep} 				% Space below theorem.
	{}						% Font name for body of theorem.
	{}						% Space to indent.
	{\bfseries}				% Font name for theorem head.
	{.}						% Punctuation between head and body.
	{\newline}				% Space after theorem head.
	{}						% Specify theorem head.

\theoremstyle{Theorem_Line_Break}

\newtheorem{theorem}[definition]{Theorem}
\newtheorem{lemma}[definition]{Lemma}
\newtheorem{corollary}[definition]{Corollary}

\makeatletter
\renewenvironment{proof}[1][\relax]{\par
  \pushQED{\qed}%
  \normalfont \topsep6\p@\@plus6\p@\relax
  \trivlist
  \item[\hskip\labelsep\itshape
    \ifx#1\relax \proofname\else\proofname{} (#1)\fi\@addpunct{.}]\ignorespaces
}{%
  \popQED\endtrivlist\@endpefalse
}
\makeatother

\usepackage[margin=1.5in]{geometry}

\title{\textbf{Uncloneable Quantum Encryption via Oracles}}

\author{
Anne Broadbent\\
\small University of Ottawa\\
\small \texttt{abroadbe@uottawa.ca}
\and
S\'ebastien Lord\\
\small University of Ottawa\\
\small \texttt{slord050@uottawa.ca}
}
\date{}

\begin{document} %%%
%%%%%%%%%%%%%%%%%%%%
%%%%%%%%%%%%%%%%%%%%

\maketitle

\begin{abstract}
Quantum information is well-known to achieve cryptographic feats that are unattainable using classical information alone.
Here, we add to this repertoire by introducing a new cryptographic functionality called \emph{uncloneable encryption}.
This functionality allows the encryption of a classical message such that two collaborating but isolated adversaries are prevented from  simultaneously recovering the message, even when the encryption key is revealed.
Clearly, such functionality is unattainable using classical information alone.

We formally define uncloneable encryption, and show how to achieve it using Wiesner's conjugate coding, combined with a quantum-secure pseudorandom function (qPRF).
Modelling the qPRF as an oracle, we show security by adapting techniques from the quantum one-way-to-hiding lemma, as well as using bounds from quantum monogamy-of-entanglement games.
\end{abstract}

%%%%%%%%%%%%%
% Main Text %
%%%%%%%%%%%%%

%======================%
\section{Introduction} %
%======================%

One of the key distinctions between classical and quantum information is given by the \emph{no-cloning principle}: unlike bits, arbitrary qubits cannot be perfectly copied \cite{Par70,WZ82,Die82}.
This principle is the basis of many of the feats of quantum cryptography, including quantum money \cite{Wie83} and quantum key distribution (QKD) \cite{BB84} (for a survey on quantum cryptography, see \cite{BS16}).

In QKD, two parties establish a shared secret key, using public quantum communication combined with an authentic classical channel.
The quantum communication allows to \emph{detect} eavesdropping: when the parties detect only a small amount of eavesdropping, they can produce a shared string that is essentially guaranteed to be private. 
Gottesman~\cite{Got03} studied \emph{quantum tamper-detection} in the case of \emph{encryption schemes}: in this work, a classical message is encrypted into a quantum ciphertext such that, at decryption time, the receiver will \emph{detect} if an adversary could have information about the plaintext when the key is revealed.
We note that classical information alone cannot produce such encryption schemes, since it is always possible to perfectly \emph{copy} ciphertexts.

Notably, Gottesman left open the question of an encryption scheme that would \emph{prevent} the \emph{splitting} of a ciphertext.
In other words, would it be possible to encrypt a classical message into a quantum ciphertext, such that no attack at the ciphertext level would be significantly successful  in producing \emph{two} quantum registers, each of which, when combined with the decryption key, could be used to reconstruct the plaintext?

In this work, we define, construct and prove security for a scheme that answers Gottesman's question in the positive. We call this \emph{uncloneable encryption}.
The core technical aspects of this work were first presented in one of the author's M.Sc.\ thesis~\cite{Lor19}.

%-------------------------------------%
\subsection{Summary of Contributions} %
\label{sec:SummaryContributions}      %
%-------------------------------------%

We consider encryption schemes that encode classical plaintexts into quantum ciphertexts, which we formalize in  \cref{definition:QECM}.
For simplicity, in this work, we consider only the one-time, symmetric-key case.
Next, we define  uncloneable encryption (\cref{definition:US}).
Informally, this can be thought of as a game, played between the honest sender (Alice) and two malicious recipients (Bob and Charlie).
First, Alice picks a message $m \in \{0,1\}^n$ and a key $k\in \{0,1\}^{\kappa(\lambda)}$  ($\kappa$ is a polynomial in some security parameter,~$\lambda$).
She encrypts her message into a quantum ciphertext register $R$.
Initially, Bob and Charlie are physically together, and they receive~$R$. They apply a quantum map to produce two registers:
Bob keeps register~$B$ and Charlie keeps register~$C$.
Bob and Charlie are then isolated.
In the next phase, Alice reveals~$k$ to both parties.
Using $k$ and their quantum register, Bob and Charlie produce $m_B$ and $m_C$ respectively.
Bob and Charlie \emph{win} if and only if $m_B = m_C=m$.
The scheme is \emph{$t$-uncloneable secure} if their winning probability is upper bounded by $2^{-n+t} + \eta(\lambda)$ for a negligible $\eta$.

Assuming that Alice picks her message uniformly at random, our results are summarized in \cref{figure:plot}, where we plot upper bounds for the winning probability of Bob and Charlie against various types of encodings, according to the length of~$m$.
First of all, if the encoding is classical, then Bob and Charlie can each keep  a copy of the ciphertext.
Combined with the key $k$, each party decrypts to obtain~$m$.
This gives the horizontal line at~$\Pr[\text{Adversaries win}] =1$.
Next,  a lower bound on the winning probability for \emph{any} encryption scheme is $\frac{1}{2^n}$ (corresponding to the parties coordinating a random guess). This is the \emph{ideal} curve.
Our goal is therefore to produce an encryption scheme that matches the ideal curve as close as possible.

It may seem that asking that Alice sample her message uniformly at random would be particularly restrictive, but this is not the case --- we show in \cref{theorem:unif->t} that security in the case of uniformly sampled messages implies security in the case of non-uniformly sampled messages, if the message size does not grow with the security parameter.
Specifically, if Bob and Charlie can win with probability at most $2^{-n + t} + \eta(\lambda)$ when the message is sampled uniformly at random, for some $t$ and some negligible function $\eta$, then they can win with probability at most $2^{-h+t} + \eta'(\lambda)$ if the message $m$ is sampled from a distribution with a min-entropy of $h$ where $\eta'$ is a negligible function which is larger than $\eta$.

Our first attempt at realizing uncloneable encryption (\cref{section:conjugate-encryption}) shows that the well-known Wiesner conjugate coding \cite{Wie83} already achieves a security bound that is better than any classical scheme.
For any strings $x, \theta \in \{0,1\}^n$, define the Wiesner state
$
	\ket{x^\theta}
	=
	H^{\theta_1} \ket{x_1}
	\tensor \ldots \tensor
	H^{\theta_n} \ket{x_n}
$.
The encryption uses a random key $r, \theta \in \{0,1\}^n$ and maps a classical message $m$ into the quantum state $\rho = \ketbra{(m \xor r)^\theta}$; given $r, \theta$, decryption consists in measuring in the basis determined by~$\theta$ to obtain~$x$ and then computing $x\oplus r$.
We sketch a proof that this satisfies a notion of security for encryption schemes.
The question of uncloneability then boils down to: ``How well can an adversary \emph{split} $\rho$ into \emph{two} registers, each of which, combined with $(\theta, r)$ can reconstruct~$m$?''
This question is answered in prior work on \emph{monogamy-of-entanglement games} \cite{TFKW13}: an optimal strategy wins with probability $\left(\frac{1}{2} + \frac{1}{2\sqrt{2}}\right)^n$. This is again illustrated in \cref{figure:plot}.

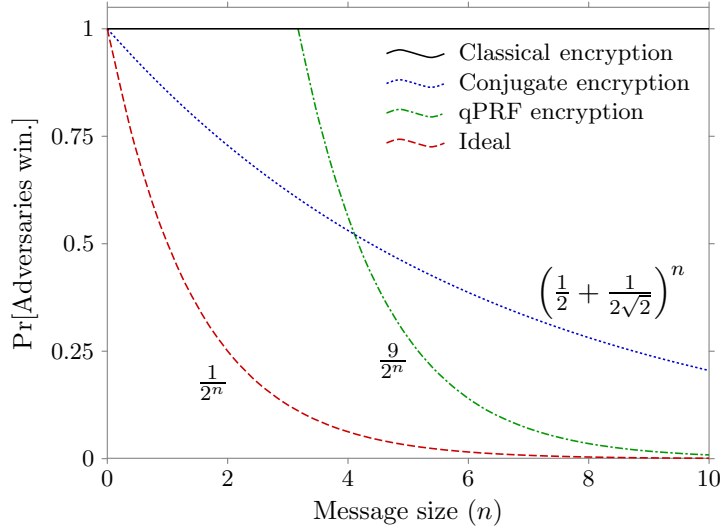
\begin{figure}
\begin{center}
\begin{tikzpicture}
	\datavisualization [
						scientific axes,
						y axis={length=6cm, ticks=few, include value=1.05, label={Pr[Adversaries win.]}},
						x axis={length=8cm, label={Message size ($n$)}},
						visualize as smooth line/.list={classical, ideal, conjugate-encryption, oracle},
						legend=north east inside,
						style sheet=strong colors,
						style sheet=vary dashing,
						classical={label in legend={text=Classical encryption}},
						conjugate-encryption={label in legend={text=Conjugate encryption}},
						oracle={label in legend={text=qPRF encryption}},
						ideal={label in legend={text=Ideal}},
						data/format=function
					   ]
		data [set=classical] {
			var x : interval [0:10];
			func y = 1;
		}
		data [set=conjugate-encryption] {
			var x : interval [0:10];
			func y = ((1/2) + 8^(-1/2))^(\value x);
		}
		data [set=oracle] {
			var x : interval [3.16993:10];
			func y = 9*2^(-\value x);
		}
		data [set=ideal] {
			var x : interval [0:10];
            func y = 2^(-\value x);
		};

	\node at (1.4,1) (eq:classical) {$\frac{1}{2^n}$};
	\node at (6.7,2.2) (eq:conjugate) {$\left(\frac{1}{2} + \frac{1}{2\sqrt{2}}\right)^n$};
	\node at (3.8, 1.3) (eq:oracle) {$\frac{9}{2^n}$};
\end{tikzpicture}
\end{center}
\caption{
	\label{figure:plot}
	Upper-bounds on winning probabilities for various types of encodings (up to negligible functions of $\lambda$) for messages sampled uniformly at random.
}
\end{figure}

In order to improve this bound, we use a quantum-secure pseudorandom function (qPRF) $f_\lambda : \{0,1\}^\lambda \times \{0,1\}^\lambda \to \{0,1\}^{n}$ (see \cref{def:qPRF}).
The encryption (see \cref{section:our-protocol}) consists of a quantum state  $\rho = \ketbra{r^\theta}$ for random $r,\theta \in \{0,1\}^\lambda$, together with a classical string  $c = m \xor f_\lambda(s, r)$ for a random~$s$.
The key $k$ consists in $\theta$ and $s$.
Once again, it can be shown that this is an encryption scheme in a more usual sense and we sketch this argument in \cref{section:our-protocol}.
Intuitively, the use of $f_\lambda$ affords us a gain in uncloneable security, because an adversary who wants to output~$m$ would need to know the pre-image of $m$ under $f_\lambda(s, \cdot)$.
Reaching a formal proof along these lines, however, is tricky.
First, we model the qPRF using a quantum oracle \cite{BBC+01,BDF+11}; this limits the adversaries' interaction with the qPRF to be black-box quantum queries.
Next, the quantum oracle model is notoriously tricky to use and many of the techniques in the classical literature are not directly applicable.
Fortunately, we can adapt techniques from Unruh's quantum one-way-to-hiding lemma~\cite{Unr15b} to the two-player setting, which enables us to recover a precise statement along the lines of the intuition above.
We thus complete the proof of our main \cref{thm:main-security-F-encryption}, obtaining the bound $9\cdot \frac{1}{2^n} + \negl(\lambda)$.
This is the fourth and final curve in \cref{figure:plot}.

In addition to the above, we formally define a different type of uncloneable security: inspired by more standard security definitions of \emph{indistinguishability}, we define \emph{uncloneable-indistinguishability} (\cref{definition:UNC-IND-security}).
This security definition bounds the advantage that the adversaries have at \emph{simultaneously} distinguishing between an encryption of $0^n$ and an encryption of a plaintext of length~$n$, as prepared by the adversaries.
In a series of results (\cref{thm:main-security-F-encryption-NE,theorem:0-UNC=>UNC-IND,thm:main-security-F-encryption-NE-IND}), we show that our main protocol achieves this security notion against adversaries that use \emph{unentangled strategies} and if the message size does not grow with $\lambda$.
As discussed in \cref{sec:applications}, there are interesting uses cases where we can assume that the adversaries do not share entanglement.

We note that our protocols (both \cref{protocol:CE} and \cref{protocol:FCE})  have the desirable property of being  \emph{prepare-and-measure} schemes.
This means that the quantum technology for the honest users is limited to the preparation of single-qubit pure states, as well as to single-qubit measurements; these quantum technologies are mature and commercially available.
(Note, however, that quantum storage remains a major challenge at the implementation level).

%-------------------------%
\subsection{Applications} %
\label{sec:applications}  %
%-------------------------%

While our focus is on the conceptual contribution of defining and proving a new primitive, we believe that uncloneable encryption could have many applications.
We give two such examples.

%:::::::::::::::::::::::::::::::::::%
\paragraph{\textbf{Quantum Money.}} %
%:::::::::::::::::::::::::::::::::::%

As it captures the idea of ``uncloneable classical information'' in a very generic manner, uncloneable encryption can be used as a tool to build other primitives which leverage the uncloneability of quantum states.
As an example, any uncloneable secure encryption scheme naturally yields a private-key quantum money scheme \cite{Wie83,AC12}.

To obtain quantum money from an uncloneable encryption scheme, we identify the notion of ``simultaneously passing the bank's verification'' with the notion of ``simultaneously obtaining the correct plaintext''. 
To generate a banknote, the bank samples a message $m$, a key $k$, a serial number $s$ and produces as output $(s, Enc(k,m))$, where $Enc(k,m)$ is the uncloneable encryption of $m$ with the key $k$.
When the bank is asked to verify a banknote, it verifies the serial number in its database to retrieve $k$, decrypts the ciphertext and verifies if the message obtained is indeed $m$.

The uncloneable security guarantee implies that the probability of a malicious party producing two banknotes which pass this test is negligible.
If this were not the case, we could use the attack which counterfeits the banknote to essentially copy the ciphertext in the underlying uncloneable encryption scheme.
The adversaries tasked with obtaining the message once the key is revealed then simply decrypt as if they were the honest receivers. 

%:::::::::::::::::::::::::::::::::::::::::::::::::::::::::::::::::::::::::%
\paragraph{\textbf{Preventing Storage Attacks by Classical Adversaries.}} %
%:::::::::::::::::::::::::::::::::::::::::::::::::::::::::::::::::::::::::%

Indistinguishable-uncloneable encryption prevents a single eavesdropping adversary with no quantum memory from collecting ciphertexts exchanged by two honest parties in the hope of later learning the key.
We sketch an argument for this fact.

Suppose such an adversary obtains a ciphertext encoded with an uncloneable-indistinguishable encryption scheme.
We claim that they cannot correctly determine if the ciphertext corresponds to the encryption of $0^n$ or of some known message $m$ with non-negligible advantage, even if the decryption key becomes known after their measurement of the ciphertext.
If such an adversary existed, it could be used to break the uncloneable-indistinguishable security of the encryption scheme.
Indeed, the almost classical eavesdropper could create two copies of their classical memory and distribute it to the two adversaries who attempt to obtain the message once the key is revealed.\footnote{We thank an anonymous reviewer for this suggestion.}

Note that the adversaries in this attack do not share any entanglement and so we can apply \cref{thm:main-security-F-encryption-NE-IND} which states that our encryption scheme is uncloneable-indistinguishable secure under this condition.

Our work is currently in the private-key setting, but can be extended in a straightforward way to the public-key setting.
In this scenario, we can still guarantee the secrecy of the message even if the eavesdropper is later able to determine the decryption key from the publicly known encryption key.
In other words, an eavesdropping adversary with no quantum memory would need to attack the ciphertext at the moment of transmission.
This is known as \emph{everlasting} security or \emph{long-term} security.

%---------------------------------%
\subsection{More on Related Work} %
\label{sec:related-work}          %
%---------------------------------%

Starting with the foundational work of Wiesner \cite{Wie83}, a rich body of literature has considered the encoding of classical information into quantum states in order to take advantage of quantum properties for cryptography.

\paragraph{\textbf{Quantum Key Recycling.}} 
The concept of quantum key recycling is a precursor to the QKD protocol, developed by Bennett, Brassard, and Breidbart~\cite{BBB14} (the manuscript was prepared in 1982 but only published recently).
According to this protocol, it is possible to encrypt a classical message into a quantum state, such that information-theoretic security is assured, but in addition, a tamper detection mechanism would allow the one-time pad key to be re-used in the case that no eavesdropping is detected.
Quantum key recycling has been the object of recent related work~\cite{DPS05,FS17}. 

\paragraph{\textbf{Tamper-Evident Encryption.}}
We referred above to tamper-detection in the case of encryption, which we will also call \emph{tamper-evident encryption}. However, we emphasize that the author originally called this contribution \emph{uncloneable encryption} \cite{Got03}.
We justify this choice of re-labelling in quoting the conclusion of the work:
\begin{quote}
One difficulty with such generalizations is that it is unclear to what extent the name ``uncloneable encryption'' is really deserved.
I have not shown that a message protected by uncloneable encryption cannot be copied --- only that Eve cannot copy it without being detected. Is it possible for Eve to create two states, (...), which can each be used (in conjunction with the secret key) to extract a good deal of information about the message?
Or can one instead prove bounds, for instance, on the sum of the information content of the various purported copies?
\cite{Got03}
\end{quote}
Since our work addresses this question, we have appropriately re-labeled prior work according to a seemingly more accurate name.

To the best of our knowledge, the precise relationship between quantum key-recycling, tamper-evident encryption, and  uncloneable encryption is unknown (see \cref{sec:future-work}).

\paragraph{\textbf{Quantum Copy-Protection.}}
Further related work includes the study of \emph{quantum copy-protection}, as initiated by Aaronson \cite{Aar09}.
Informally, this is a means to encode a function (from a given family of functions) into a quantum program state, such that an honest party can evaluate the function given the program state, but it would be impossible to somehow \emph{split} the quantum program state so as to enable \emph{two} parties to simultaneously evaluate the function.
Aaronson gave protocols for quantum copy-protection in an oracle model, but left wide open the question of quantum copy-protection in the plain model.
In a way, uncloneable encryption is a first step towards quantum copy-protection, since it prevents copying of \emph{data}, which can be seen as a unit of information that is even simpler than a function.

%------------------------------------%
\subsection{Outlook and Future Work} %
\label{sec:future-work}              %
%------------------------------------%

In this work, we show that, thanks to quantum information, one of the basic tacit assumptions of encryption, namely that an adversary can copy ciphertexts, is challenged.
We believe that this has the potential to significantly change the landscape of cryptography, for instance in  terms of techniques for \emph{key management}~\cite{NIST-SP-800-57-P1}.
Furthermore, our techniques could become building blocks for a theory of uncloneable cryptography.

Our work leads to many follow-up questions, broadly classified according to the following themes: 

\paragraph{\textbf{Improvements.}}
There are many possible improvements to the current work.
For instance: Could our scheme be made resilient to errors?
Can we remove the reliance on the oracle, and/or on the qPRF?
Could an encryption scheme simultaneously be uncloneable \emph{and} provide \emph{tamper detection}?
Would achieving uncloneable-indistinguishable security be possible, without any restrictions on the adversary's strategy?

\paragraph{\textbf{Links with related work.}}
What are the links, if any, between uncloneable encryption, tamper-evident encryption~\cite{Got03}, and quantum encryption with key recycling~\cite{BBB14,DPS05,FS17}?
We note that both uncloneable encryption and quantum encryption with key recycling~\cite{FS17} make use of theorems developed in the context of one-sided device-independent QKD~\cite{TFKW13}.
Can we make more formal links between these primitives?

\paragraph{\textbf{More uncloneability.}}
Finally, our work paves the way for the study of more complex unclonable primitives.
Could this lead to uncloneable programs~\cite{Aar09}?
What about in complexity theory, could we define and realize uncloneable \emph{proofs}~\cite{Aar09}?

%--------------------%
\subsection{Outline} %
%--------------------%

The remainder of the paper is structured as follows.
In \cref{sec:prelims}, we introduce some basic notation and useful results from the literature.
In \cref{sec:uncloneable encryption}, we formally define uncloneable encryption schemes and their security.
Our two protocols are described and proved secure in~\cref{sec:protocols}.

%=======================%
\section{Preliminaries} %
\label{sec:prelims}     %
%=======================%

In this section, we present basic notation, together with  techniques from prior work that are used in the remainder of the paper.

%-------------------------------------------------------%
\subsection{Notation and Basics of Quantum Information} %
%-------------------------------------------------------%

We denote the set of all functions of the form $f : \{0,1\}^n \to \{0,1\}^m$ by $\Bool(n,m)$.
We denote the set of strictly positive natural numbers by $\N^+$.
All Hilbert spaces are finite dimensional.
We overload the expectation symbol $\E$ in the following way:
If $X$ is a finite set, $\mc{X}$ a random variable on $X$, and $f : X \to \R$ some function, we define
\begin{equation}
	\E_{x \gets \mc{X}} f(x) = \sum_{x \in X} \Pr\left[x = \mc{X}\right] f(x).
\end{equation}
If we omit the random variable, we assume a uniform distribution.
In other words, $\E_x f(x) = \frac{1}{\abs{X}} \sum_{x \in X} f(x)$.

A comprehensive introduction to quantum information and quantum computing may be found in \cite{NC00,Wat18}.
We fix some notation in the following paragraphs.

Let $\Qb = \C^2$ be the state space of a single qubit.
In particular, $\Qb$ is a two-dimensional complex Hilbert space spanned by the orthonormal set $\{\ket{0}, \ket{1}\}$.
For any $n \in \N^+$, we write $\Qb(n) = \Qb^{\tensor n}$ and note that
\begin{equation}
\{\ket{s} = \ket{s_1} \tensor \ket{s_2} \tensor \ldots \tensor \ket{s_n}\}_{s \in \{0,1\}^n}
\end{equation}
forms an orthonormal basis of $\Qb(n)$.

Let $\HS$ be a Hilbert space.
The set of all unitary and density operators on $\HS$ are denoted by $\mc{U}(\HS)$ and $\mc{D}(\HS)$, respectively.
We recall that the operator norm of any linear operator $A : \mc{H} \to \mc{H'}$ between finite dimensional Hilbert spaces is given by
\begin{equation}
	\norm{A} = \max_{\substack{v \in \mc{H} \\ \norm{v} = 1}} \norm{A v}
\end{equation}
and satisfies the property that $\norm{Av} \leq \norm{A} \cdot \norm{v}$.
If $A$ is either a projector or a unitary operator, then $\norm{A} = 1$.

We use the term ``quantum state'' to refer to both unit vectors $\ket{\psi} \in \mc{H}$
and to density operators $\rho \in \mc{D}(\mc{H})$ on some Hilbert space.

If $H \in \mc{U}(\Qb)$ is the Hadamard operator defined by
\begin{equation}
	\ket{0} \mapsto \frac{\ket{0} + \ket{1}}{\sqrt{2}}
	\text{ and }
	\ket{1} \mapsto \frac{\ket{0} - \ket{1}}{\sqrt{2}}
\end{equation}
then, for any strings $x, \theta \in \{0,1\}^n$, we define
\begin{equation}
	\ket{x^\theta}
	=
	H^{\theta_1} \ket{x_1}
	\tensor
	H^{\theta_2} \ket{x_2}
	\tensor \ldots \tensor
	H^{\theta_n} \ket{x_n}
\end{equation}
and note that $\left\{\ket{s^\theta}\right\}_{s \in \{0,1\}^n}$ forms an orthonormal basis of $\Qb(n)$.
Following their prominent use in \cite{Wie83}, we call states of the form $\ket{x^\theta}$ Wiesner states and, for any fixed $\theta \in \{0,1\}^n$, we call $\left\{\ket{s^\theta}\right\}_{s \in \{0,1\}^n}$ a Wiesner basis.

For any $n \in \N^+$, we define the Einstein-Podolski-Rosen \cite{EPR35} (EPR) state by
\begin{equation}
	\ket{\text{EPR}_n} = \frac{1}{\sqrt{2^n}} \sum_{x \in \{0,1\}^n} \ket{x}\ket{x}
\end{equation}
and note that it is an element of $\mc{Q}(2n)$.

A positive operator-valued measurement (POVM) on a Hilbert space $\mc{H}$ is a finite collection of positive semidefinite operators $\{E_i\}_{i \in I}$ on $\mc{H}$ which sum to the identity.
A projective measurement is a POVM composed of projectors.

We also recall that physically permissible transformation of a quantum system precisely coincide with the set of completely positive trace preserving (CPTP) maps.
In particular, CPTP map will map density operators to density operators.

A polynomial-time uniform family of circuits $\texttt{C} = \{\texttt{C}_\lambda\}_{\lambda \in \N^+}$ is a collection of quantum circuits indexed by $\N^+$ such that there exists a polynomial-time deterministic Turing machine $\texttt{T}$ which, on input  $1^\lambda$, produces a description of $\texttt{C}_\lambda$.
We refer to such families as efficient circuits.
Each circuit $\texttt{C}_\lambda$ defines and implements a certain CPTP map $C_\lambda : \mc{D}(\HS_{\text{In}, \lambda}) \to \mc{D}(\HS_{\text{Out}, \lambda})$, where the Hilbert spaces $\HS_{\text{In}, \lambda}$ and $\HS_{\text{Out}, \lambda}$ are implicitly defined by the circuit.
Note that we consider general, which is to say possibly non-unitary, circuits.
These were introduced in \cite{AKN98}.
It is worth noting that a universal gate set for general quantum circuits exists which is composed of only unitary gates, implementing maps of the form $\rho \mapsto U \rho U^\dag$ for some unitary operator $U$, and two non-unitary maps which are
\begin{itemize}
	\item
		the single qubit partial trace map $\Tr : \mc{D}(\Qb) \to \mc{D}(\C)$ and
	\item
		the state preparation map $\text{Aux} : \mc{D}(\C) \to \mc{D}(\Qb)$ defined by $1 \mapsto \ketbra{0}$.
\end{itemize}
Further information on this circuit model can be found in \cite{Wat09}.

%-------------------------------------------%
\subsection{Monogamy of Entanglement Games} %
%-------------------------------------------%

Monogamy-of-entanglement games were introduced and studied in \cite{TFKW13}.
In short, a monogamy-of-entanglement game is played by Alice against cooperating Bob and Charlie.
Alice describes to Bob and Charlie a collection of different POVMs which she could use to measure a quantum state on a Hilbert space $\mc{H}_A$.
These POVMs are indexed by a finite set $\Theta$ and each reports a measurement result taken from a finite set $X$.
Bob and Charlie then produce a tripartite quantum state $\rho \in \mc{D}(\mc{H}_A \tensor \mc{H}_B \tensor \mc{H}_C)$, giving the $A$ register to Alice, the $B$ register to Bob and the $C$ register to Charlie.
Alice then picks a $\theta \in \Theta$, measures her subsystem with the corresponding POVM and obtains some result $x \in X$.
She then announces~$\theta$ to Bob and Charlie who are now isolated.
Bob and Charlie win if and only if they can both simultaneously guess the result $x$.

Upper bounds on the winning probability of Bob and Charlie in such games was the primary subject of study in \cite{TFKW13}.
One of their main results, corresponding to a game where Alice measures in a random Wiesner basis, is as follows.

\begin{theorem}
\label{theorem:TFKW13}
Let $\lambda \in \N^+$.
For any Hilbert spaces $\HS_B$ and $\HS_C$, any collections of POVMs
\begin{equation}
	\left\{
		\left\{B^\theta_x\right\}_{x \in \{0,1\}^\lambda}
	\right\}_{\theta \in \{0,1\}^n}
	\text{ and }
	\left\{
		\left\{C^\theta_x\right\}_{x \in \{0,1\}^\lambda}
	\right\}_{\theta \in \{0,1\}^n}
\end{equation}
on these Hilbert spaces, and any state $\rho \in \mc{D}(\Qb(\lambda) \tensor \HS_B \tensor \HS_C)$, we have that
\begin{equation}
	\E_\theta
	\sum_{x \in \{0,1\}^\lambda}
	\Tr\left[
		\left(\ketbra{x^\theta} \tensor B_x^\theta \tensor C_x^\theta\right)
		\rho
	\right]
	\leq
	\left(\frac{1}{2} + \frac{1}{2\sqrt{2}}\right)^\lambda.
\end{equation}
\end{theorem}

Using standard techniques, we can recast this theorem in a context where Alice sends to Bob and Charlie a random Wiesner state and Bob and Charlie split this state among themselves via some CPTP map $\Phi$.

\begin{corollary}
\label{theorem:TFKW13-ours}
Let $\lambda \in \N^+$.
For any Hilbert spaces $\HS_B$ and $\HS_C$, any collections of POVMs
\begin{equation}
	\left\{
		\left\{B^\theta_x\right\}_{x \in \{0,1\}^\lambda}
	\right\}_{\theta \in \{0,1\}^\lambda}
	\text{ and }
	\left\{
		\left\{C^\theta_x\right\}_{x \in \{0,1\}^\lambda}
	\right\}_{\theta \in \{0,1\}^\lambda}
\end{equation}
on these Hilbert spaces, and any CPTP map $\Phi : \mc{D}(\Qb(\lambda)) \to \mc{D}(\HS_B \tensor \HS_C)$, we have that
\begin{equation}
	\E_\theta
	\E_x
	\Tr\left[
		\left(B_x^\theta \tensor C_x^\theta\right)
		\Phi\left(\ketbra{x^\theta}\right)
	\right]
	\leq
	\left(\frac{1}{2} + \frac{1}{2\sqrt{2}}\right)^\lambda
	.
\end{equation}
\end{corollary}

The proof is relegated to \cref{section:proofs}, but conceptually follows from a two-step argument.
First, we only consider states of the form $\left(\Id \tensor \Phi\right)\ketbra{\text{EPR}_\lambda}$ for some CPTP map $\Phi$ and where Alice keeps the intact subsytems from the EPR pairs.
Then, we apply the correspondence between Alice measuring her half of an EPR pair in a random Wiesner basis and her sending a random Wiesner state.
This correspondence is similar to the one used in the Shor-Preskill proof of security for the BB84 QKD protocol \cite{SP00}.

\cref{theorem:TFKW13-ours} can be seen as the source of ``uncloneability'' for our upcoming protocols.
When Alice sends a state $\ketbra{x^\theta}$, picked uniformly at random, to Bob and Charlie, she has a guarantee that it is unlikely for both of them to learn~$x$ even if she later divulges~$\theta$.

It is worth noting that \cref{theorem:TFKW13} and \cref{theorem:TFKW13-ours} have no computational or hardness assumptions.
This makes them an ideal tool with which to build uncloneable encryption schemes.

%--------------------------------------------------------------%
\subsection{Oracles and Quantum-Secure Pseudorandom Functions} %
%--------------------------------------------------------------%
\label{section:oracles}

A quantum-secure pseudorandom function is a keyed function which appears random to an efficient quantum adversary who only sees its input/output behaviour and is ignorant of the particular key being used.
We formally define this notion with the help of oracles.
Quantum accessible oracles have been previously studied in the literature, for example in \cite{BBC+01,BDF+11,Unr15b}.

Given a function $H \in \Bool(n,m)$, a quantum circuit $\texttt{C}$ is said to have oracle access to $H$, denoted $\texttt{C}^H$, if we add to its gate set a gate implementing the unitary operator $O^H \in \mc{U}(\Qb(n)_Q \tensor \Qb(m)_R)$ defined on computational basis states by
\begin{equation}
	\ket{x}_Q \tensor \ket{y}_R \mapsto \ket{x}_Q \tensor \ket{y \xor H(x)}_R\,.
\end{equation}
Colloquially, we are giving $\texttt{C}$ a ``black box'' which computes the function $H$.
Note that if $H, H' \in \Bool(n,m)$ are two functions, we can obtain the circuit $C^{H'}$ from $C^H$ by replacing every instance of the $O^H$ gate by the $O^{H'}$ gate.

We can now give a definition, inspired by the one in \cite{Zha12}, of a quantum-secure pseudorandom function.

\begin{definition}[Quantum-Secure Pseudorandom Function]
\label{def:qPRF}
A \emph{quantum-secure pseudorandom function} $\mc{F}$ is a collection of functions
\begin{equation}
	\mc{F}
	=
	\left\{
		f_\lambda:
		\{0,1\}^\lambda \times\{0,1\}^{\ell_\text{In}(\lambda)}
		\to \{0,1\}^{\ell_\text{Out}(\lambda)}\right\}_{\lambda \in \N^+}
\end{equation}
where $\ell_\text{In}, \ell_\text{Out} : \N^+ \to \N^+$ and such that:
\begin{enumerate}
	\item
		There is an efficient quantum circuit $\texttt{F} = \{\texttt{F}_\lambda\}_{\lambda \in \N^+}$ such that $\texttt{F}_\lambda$ implements the CPTP map $F_\lambda(\rho) = U_\lambda \rho U_\lambda^\dag$ where $U_\lambda \in \mc{U}(\mc{Q}(\lambda + \ell_\text{In}(\lambda) + \ell_\text{Out}(\lambda)))$ is defined by
		\begin{equation}
			U_\lambda \big(\ket{k} \ket{a} \ket{b}\big) = \ket{k} \ket{a} \ket{b \xor f_\lambda(k, a)}\,.
		\end{equation}
	\item
		For all efficient quantum circuits $\texttt{D} = \{\texttt{D}^H_\lambda\}_{\lambda \in \N^+}$ having oracle access to a function of the form $H \in \Bool(\ell_\text{In}(\lambda), \ell_\text{Out}(\lambda))$, each implementing a CPTP map of the form $D^H_\lambda : \mc{D}(\C) \to \mc{D}(\mc{Q})$, there is a negligible function $\eta$ such that:
		\begin{equation}
			\abs{
				\E_{k}
				\Tr\big[
					\ketbra{0} D^{f_\lambda(k, \cdot)}_\lambda(1)
				\big]
				-
				\E_{H}
				\Tr\big[
					\ketbra{0} D^{H}_\lambda(1)
				\big]
			}
			\leq
			\eta(\lambda)\,.
		\end{equation}
\end{enumerate}
\end{definition}

We should think of $\texttt{D}$ as a circuit which attempts to distinguish two different cases: is it given oracle access to an instance of the pseudorandom function, which is to say $f(k, \cdot) : \{0,1\}^{\ell_\text{In}(\lambda)} \to \{0,1\}^{\ell_\text{Out}(\lambda)}$ for a randomly sampled $k \in \{0,1\}^\lambda$?
Or to a function that was sampled truly at random, $H \in \Bool(\ell_\text{In}(\lambda), \ell_\text{Out}(\lambda))$?

The circuit takes no input and produces a single bit of output, via measuring a single qubit in the computational basis.
The bound given in the definition ensures that the probability distribution of the output does not change by much in both scenarios.

In his work on quantum-secure pseudorandom functions \cite{Zha12}, Zhandry showed that certain pseudorandom functions that are secure against classical adversaries are insecure against quantum adversaries.
Fortunately, Zhandry also showed that some common constructions of pseudorandom functions remain secure against quantum adversaries.

%==================================%
\section{Uncloneable Encryption}   %
\label{sec:uncloneable encryption} %
%==================================%

The encryption of classical plaintexts into classical ciphertexts has been extensively studied.
The study of encrypting quantum plaintexts into quantum ciphertexts has also received some attention, for example in \cite{ABF+16}.
Uncloneable encryption is a security notion for classical plaintexts which is impossible to achieve in any meaningful way with classical ciphertexts.
Thus,  we need to formally define a notion of quantum encryptions for classical messages.

%------------------------------------------------------%
\subsection{Quantum Encryptions of Classical Messages} %
%------------------------------------------------------%

A quantum encryption of classical messages scheme is a procedure which takes as input a plaintext and a key, in the form of classical bit strings, and produces a ciphertext in the form of a quantum state.
We model these schemes as efficient quantum circuits and CPTP maps where classical bit strings are identified with computational basis states: $s \leftrightarrow \ketbra{s}$.
Our schemes are parametrized by a security parameter $\lambda$.
In general, the message size $n = n(\lambda)$, the key size $\kappa = \kappa(\lambda)$, and the size of the ciphertext $\ell = \ell(\lambda)$ may depend on $\lambda$.
This is formalized in \cref{definition:QECM}.

\begin{definition}[Quantum Encryption of Classical Messages]
\label{definition:QECM}
A \emph{quantum encryption of classical messages} (QECM) scheme is a triplet of efficient quantum circuits $\mc{S} = (\texttt{Key}, \texttt{Enc}, \texttt{Dec})$ implementing CPTP maps of the form
\begin{itemize}
	\item
		$Key_\lambda : \mc{D}(\C) \to \mc{D}(\HS_{K,\lambda})$,
	\item
		$Enc_\lambda : \mc{D}(\HS_{K, \lambda} \tensor \HS_{M,\lambda}) \to \mc{D}(\HS_{T,\lambda})$, and
	\item
		$Dec_\lambda : \mc{D}(\HS_{K, \lambda}  \tensor \HS_{T, \lambda}) \to \mc{D}(\HS_{M,\lambda})$
\end{itemize}
where $\HS_{M, \lambda} = \Qb(n(\lambda))$ is the plaintext space, $\HS_{T, \lambda} = \Qb(\ell(\lambda))$ is the ciphertext space, and $\HS_{K, \lambda} = \Qb(\kappa(\lambda))$ is the key space for functions $n, \ell, \kappa : \N^+ \to \N^+$.

For all $\lambda \in \N^+$, $k \in \{0,1\}^{\kappa(\lambda)}$, and $m \in \{0,1\}^{n(\lambda)}$, the maps must satisfy
\begin{equation}
\label{equation:QECM-cc}
	\Tr[\ketbra{k} Key(1)] > 0
	\implies
	\Tr[\ketbra{m} Dec_k \circ Enc_k \ketbra{m}] = 1
\end{equation}
where $\lambda$ is implicit, $Enc_k$ is the CPTP map defined by $\rho \mapsto Enc(\ketbra{k} \tensor \rho)$, and we define $Dec_k$ analogously.
\end{definition}

A short discussion on the key generation circuit, $\texttt{Key}$, is in order.
First, note that $\texttt{Key}$ takes no input.
Indeed, the domain of $Key_\lambda$ is $\mc{D}(\C)$ and $\C$ is the state space of zero qubits.
In particular, there is a single valid quantum state on $\C$: $\mc{D}(\C) = \{1\}$.
To generate a classical key to be used by the encryption and decryption circuits $\texttt{Enc}_\lambda$ and $\texttt{Dec}_\lambda$, a party runs the circuit $\texttt{Key}_\lambda$ and obtains the quantum state $Key_\lambda(1)$.
This quantum state is then measured in the computational basis and the result of this measurement is used as the key.
We then see that \cref{equation:QECM-cc} is a correctness condition which imposes that, for all keys that may be generated, a valid ciphertext is always correctly decrypted.

%-----------------------------%
\subsection{Security Notions} %
%-----------------------------%

Now that we have formal definition for QECM schemes, we can define security notions for these schemes.
We define three such notions:
\begin{enumerate}
	\item
		Indistinguishable security.
		Conceptually inspired by the original security notion of indistinguishable encryptions \cite{GM84}, which considers classical plaintexts and classical ciphertexts, and similar in details to an analogue definition in \cite{ABF+16} which considers quantum plaintexts and quantum ciphertexts, this security notion considers classical plaintexts and quantum ciphertexts.
		It is formally stated in \cref{definition:IND-security}.
	\item
		Uncloneable security.
		This security notion is novel to this work and captures, in the broadest sense, what we mean by an ``uncloneable encryption scheme''.
		This security notion is defined in \cref{definition:UNC-security} and is paramatrized by a real value $0 \leq t \leq n$, where $n$ is the message size.
		The case where $t = 0$ is ideal and $t = n$ is trivial.
		In particular, no encryption scheme with classical ciphertexts may achieve $t$-uncloneable security for $t < n$.
	\item
		Uncloneable-indistinguishable security.
		This security notion is also novel to this work.
		It can be seen as a combination of indistinguishable and uncloneable security.
		It is formally defined in \cref{definition:UNC-IND-security}.
\end{enumerate}

Each of these security notions is defined in two steps.
First, we define a type of attack (\cref{definition:IND-attack,definition:UNC-attack,definition:UNC-IND-attack}).
Then, we say that the QECM scheme achieves the given security notion if all admissible attacks have their winning probability appropriately bounded (\cref{definition:IND-security,definition:UNC-security,definition:UNC-IND-security}).
The definitions for uncloneable security and uncloneable-indistinguishable security will formalize the games which we described in \cref{sec:SummaryContributions}.

Note that many classical encryption schemes which are secure against quantum adversaries, such as the one-time pad, are indistinguishable secure but satisfy neither uncloneable security notions as their ciphertexts can alway be perfectly copied.
We also discuss in \cref{section:conjugate-encryption} a scheme which offers non-trivial uncloneable security but is not in any way uncloneable-indistinguishable secure.

We first define our notion of indistinguishable security.

\begin{definition}[Distinguishing Attack]
\label{definition:IND-attack}
Let $\mc{S}$ be a QECM scheme.
A \emph{distinguishing attack} against $\mc{S}$ is a pair of efficient quantum circuits $\mc{A} = \left(\texttt{G}, \texttt{A}\right)$ implementing CPTP maps of the form
\begin{itemize}
	\item
		$G_\lambda : \mc{D}(\C) \to \mc{D}(\mc{H}_{S, \lambda} \tensor \mc{H}_{M, \lambda})$ and
	\item
		$A_\lambda : \mc{D}(\mc{H}_{S,\lambda} \tensor \mc{H}_{T, \lambda}) \to \mc{D}(\mc{Q})$
\end{itemize}
where $\mc{H}_{S, \lambda} = \mc{Q}(s(\lambda))$ for a function $s : \N^+ \to \N^+$ and $\mc{H}_{M, \lambda}$ and $\mc{H}_{T, \lambda}$ are as defined in $\mc{S}$.
\end{definition}

\begin{definition}[Indistinguishable Security]
\label{definition:IND-security}
Let $\mc{S}$ be an QECM scheme.
For a fixed and implicit $\lambda$, we define the CPTP map $Enc_{k}^1 : \mc{D}(\mc{H}_{M, \lambda}) \to \mc{D}(\mc{H}_{T, \lambda})$ by
\begin{equation}
	\rho
	\mapsto
	\sum_{m \in \{0,1\}^n}
	\Tr\left[
		\ketbra{m} \rho
	\right]
	\cdot
	Enc_{k}(\ketbra{m})
\end{equation}
and the CPTP map $Enc_k^0 : \mc{D}(\mc{H}_{M, \lambda}) \to \mc{D}(\mc{H}_{T, \lambda})$ by
\begin{equation}
	\rho \mapsto Enc_k(\ketbra{0^n})
\end{equation}
where $0^n \in \{0,1\}^n$ is the all zero bit string.

Then, we say that $\mc{S}$ is \emph{indistinguishable secure} if for all distinguishing attacks $\mc{A}$ against $\mc{S}$ there exists a negligible function $\eta$ such that
\begin{equation}
	\E_b
	\E_{k \gets \mc{K}}
	\Tr\left[
		\ketbra{b}
		A_\lambda
		\circ 
		\left(
			\Id_S \tensor Enc_k^b
		\right)
		\circ
		G(1)
	\right]
	\leq
	\frac{1}{2}
	+
	\eta(\lambda)
\end{equation}
where $\lambda$ is implicit on the left-hand side, $b \in \{0,1\}$,  and $\mc{K_\lambda}$ is the random variable distributed on $\{0,1\}^{\kappa(\lambda)}$ such that $\Pr\left[\mc{K}_\lambda = k\right]=\Tr\left[\ketbra{k} Key_\lambda(1)\right]$.
\end{definition}

In \cref{definition:IND-security}, the map $Enc_k^0$ should be seen as discarding whatever plaintext was given and producing the encryption of the all zero bit string.
On the other hand, $Enc_k^1$ is the map which first measures the state given in the computational basis, to ensure that the plaintext is indeed a classical message, and then encrypts this message.
We say that a QECM scheme has indistinguishable security if no efficient adversary can distinguish between both of these scenarios with more then a negligible advantage.
This security notion allows us to show that the schemes we define do offer a level of security as encryption schemes.

Next, we formalize the intuitive definition for uncloneable security as given by the game described in \cref{sec:SummaryContributions}.
In \cref{figure:UA}, we sketch out the relation between the various CPTP maps and the underlying Hilber spaces considered in this definition.

\begin{definition}[Cloning Attack]
\label{definition:UA}
\label{definition:UNC-attack}
Let $\mc{S}$ be a QECM scheme.
A \emph{cloning attack} against $\mc{S}$ is a triplet of efficient quantum circuits $\mc{A} = (\texttt{A}, \texttt{B}, \texttt{C})$ implementing CPTP maps of the form
\begin{itemize}
	\item
		$A_\lambda : \mc{D}(\HS_{T, \lambda}) \to \mc{D}(\HS_{B, \lambda} \tensor \HS_{C, \lambda})$,
	\item
		$B_\lambda : \mc{D}(\HS_{K, \lambda} \tensor \HS_{B, \lambda}) \to \mc{D}(\mc{H}_{M, \lambda})$, and
	\item
		$C_\lambda : \mc{D}(\HS_{K, \lambda} \tensor \HS_{C, \lambda}) \to \mc{D}(\mc{H}_{M, \lambda})$
\end{itemize}
where $\mc{H}_{B, \lambda} = \Qb(\beta(\lambda))$ and $\mc{H}_{C, \lambda} = \Qb(\gamma(\lambda))$ for some functions $\beta, \gamma : \N^+ \to \N^+$ and $\HS_{K, \lambda}$, $\HS_{M, \lambda}$, and $\HS_{T, \lambda}$ are as defined by $\mc{S}$.
\end{definition}

\begin{figure}
\begin{center}
\begin{tikzpicture}[x=1cm,y=1cm]

	\node (HMI) at (0, 0)   {$\mc{H}_M$};
	\node (HT)  at (2, 0)   {$\mc{H}_T$};
	\node (HB)  at (5, 0.5) {$\mc{H}_B$};
	\node (HC)  at (5,-0.5) {$\mc{H}_C$};
	\node (HMB) at (7, 0.5) {$\mc{H}_M$};
	\node (HMC) at (7,-0.5) {$\mc{H}_M$};

	\draw[->] (HMI)   -- node [midway,above] {$Enc_k$} (HT);
	\draw     (HT)    -- node [midway,above] {$A$}   (3.5,0);
	\draw[->] (3.5,0) -- (HB);
	\draw[->] (3.5,0) -- (HC);
	\draw[->] (HB)    -- node [midway,above] {$B_k$} (HMB);
	\draw[->] (HC)    -- node [midway,above] {$C_k$} (HMC);	

\end{tikzpicture}
\end{center}
\label{figure:UNC-Attack}
\caption{
	\label{figure:UA}
	Schematic representation of the maps considered in a cloning attack as given in \cref{definition:UNC-attack}.
}
\end{figure}
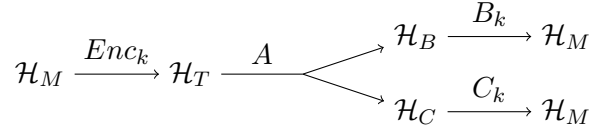

\begin{definition}[Uncloneable Security]
\label{definition:US}
\label{definition:UNC-security}
A QECM scheme $\mc{S}$ is \emph{$t(\lambda)$-uncloneable secure} if for all cloning attacks $\mc{A}$ against $\mc{S}$ there exists a negligible function $\eta$ such that
\begin{equation}
	\E_{m}
	\E_{k \gets \mc{K}}
	\Tr\left[
		\left(\ketbra{m} \tensor \ketbra{m}\right)
		\left(B_k \tensor C_k\right) \circ A \circ Enc_k
		\left(\ketbra{m}\right)
	\right]
	\leq
	2^{-n+t(\lambda)} + \eta(\lambda)
\end{equation}
where $\lambda$ is implicit on the left-hand side, $\mc{K}_\lambda$ is a random variable distributed on $\{0,1\}^{\kappa(\lambda)}$ such that $\Pr\left[\mc{K}_\lambda = k\right] = \Tr\left[\ketbra{k} Key_\lambda(1)\right]$ and $B_k$ is the CPTP map defined by $\rho \mapsto B(\ketbra{k} \tensor \rho)$ and similarly for $C_k$.

If $\mc{S}$ is $0$-uncloneable secure, we say that it is simply \emph{uncloneable secure}.
\end{definition}

We note that any encryption which produces classical ciphertexts cannot be $t$-uncloneable secure for any $t < n$.
Indeed, an attack $\mc{A}$ where $\texttt{A}$ copies the classical ciphertext and where $\texttt{B} = \texttt{C} = \texttt{Dec}$ succeeds with probability $1$.

Our definition of uncloneable security is with respect to messages sampled uniformly at random.
However, if the length of the message is fixed, $t$-uncloneable security implies a similar security notion for messages sampled according to other distributions.
We formalize this in the next theorem.

\begin{theorem}
\label{theorem:unif->t}
Let $\mc{S}$ be a QECM scheme which is $t$-uncloneable secure and whose message size is constant, i.e.: $n(\lambda) = n$.
Let $\mc{M}$ be a random variable distributed over $\{0,1\}^{n}$ with min-entropy $h$.
Then, for any cloning attack $\mc{A}$ on $\mc{S}$ there is a negligible function $\eta$ such that
\begin{equation}
	\E_{m \gets \mc{M}}
	\E_{k \gets \mc{K}}
	\Tr\left[
		\left(\ketbra{m} \tensor \ketbra{m}\right)
		\left(B_k \tensor C_k\right)
		\circ
		A
		\circ
		Enc_k
		\ketbra{m}
	\right]
	\leq
	2^{-h+t(\lambda)} + \eta(\lambda)
\end{equation}
where $\lambda$ is implicit on the left-hand side.
\end{theorem}
\begin{proof}
For all $k \in \{0,1\}^{\kappa(\lambda)}$ and $m \in \{0,1\}^{n}$, define
\begin{equation}
	p(k, m)
	=
	\Tr\left[
		\left(\ketbra{m} \tensor \ketbra{m}\right)
		\left(B_k \tensor C_k\right)
		\circ
		A
		\circ
		Enc_k
		\left(\ketbra{m}\right)
	\right].
\end{equation}
Recalling the min-entropy of $\mc{M}$ and that $\mc{S}$ is $t$-uncloneable, we may write
\begin{align}
	&
	\E_{m \gets \mc{M}}
	\E_{k \gets \mc{K}}
	\Tr\left[
		\left(\ketbra{m} \tensor \ketbra{m}\right)
		\left(B_k \tensor C_k\right) \circ A \circ Enc_k
		\left(\ketbra{m}\right)
	\right]
	\\
	&=
	\sum_{m \in \{0,1\}^n} \Pr\left[\mc{M} = m\right] \E_{k \gets \mc{K}} p(k, m)
	\leq
	2^{-h} \cdot 2^n
	\E_{m}
	\E_{k \gets \mc{K}} p(k,m)
	\leq
	2^{-h}\left(2^{t} + 2^n \eta(\lambda)\right). \nonumber
\end{align}
Noting that $\lambda \mapsto 2^{-h+n}\eta(\lambda)$ is a negligible function concludes the proof.
\end{proof}

Finally, we formalize the notion of uncloneable-indistinguishable security (see \cref{sec:SummaryContributions} for a description in terms of a game, and \cref{figure:UNC-IND-attack} for the relation between the various CPTP maps and the underlying Hilbert spaces).

\begin{definition}[Cloning-Distinguishing Attack]
\label{definition:UNC-IND-attack}
Let $\mc{S}$ be a QECM scheme.
A \emph{cloning-distinguishing attack} against $\mc{S}$ is a tuple $\mc{A} = \left(\texttt{G}, \texttt{A}, \texttt{B}, \texttt{C}\right)$ of efficient quantum circuits implementing CPTP maps of the form
\begin{itemize}
	\item
		$G_\lambda : \mc{D}(\C) \to \mc{D}(\mc{H}_{S, \lambda} \tensor \mc{H}_{M, \lambda})$,
	\item
		$A_\lambda : \mc{D}(\mc{H}_{S, \lambda} \tensor \mc{H}_{T, \lambda}) \to \mc{D}(\mc{H}_{B, \lambda} \tensor \mc{H}_{C, \lambda})$,
	\item
		$B_\lambda : \mc{D}(\mc{H}_{K, \lambda} \tensor \mc{H}_{B, \lambda}) \to \mc{D}(\mc{Q})$, and
	\item
		$C_\lambda : \mc{D}(\mc{H}_{K, \lambda} \tensor \mc{H}_{C, \lambda}) \to \mc{D}(\mc{Q})$
\end{itemize}
where $\mc{H}_{S,\lambda} = \mc{Q}(s(\lambda))$, $\mc{H}_{B, \lambda} = \mc{Q}(\beta(\lambda))$, and $\mc{H}_{C, \lambda} = \mc{Q}(\alpha(\lambda))$ for some functions $s, \alpha, \beta : \N^+ \to \N^+$ and all other Hilbert spaces are as defined by $\mc{S}$.
\end{definition}

\begin{figure}[b] % Manual layout correction.
\begin{center}
\begin{tikzpicture}[x=1cm,y=1cm]

	\node (I)   at (0,  0) {$\C$};
	\node (N1)  at (2,  0) {};
	\node (HM)  at (3, -0.5) {$\mc{H}_M$};
	\node (N2)  at (3,  0.5) {};
	\node (HT)  at (5, -0.5) {$\mc{H}_T$};
	\node (HS)  at (4,  0.5) {$\mc{H}_S$};
	%\node (Q)   at (8,  0) {$\mc{Q}$};
	\node (HB)  at (9,  0.5) {$\mc{H}_B$};
	\node (HC)  at (9, -0.5) {$\mc{H}_C$};
	\node (HMB) at (11, 0.5) {$\mc{Q}$};
	\node (HMC) at (11,-0.5) {$\mc{Q}$};

	\draw     (I)    -- node [midway,above] {$G$} (2,0);
	\draw[->] (2,0)  -- (HM);
	\draw[->] (HM)   -- node [midway,above] {$Enc^b_k$} (HT);
	\draw     (2,0)  -- (3,0.5);
	\draw[->] (3,0.5)  -- (HS);
	\draw     (HT)   -- (6,0);
	\draw     (6,0)  -- node [midway,above] {$A$} (8, 0);
	\draw     (HS)   -- (5,0.5);
	\draw     (5,0.5)  -- (6,0);
	\draw[->] (8,0)  -- (HB);
	\draw[->] (8,0)  -- (HC);
	\draw[->] (HB)   -- node [midway,above] {$B_k$} (HMB);
	\draw[->] (HC)   -- node [midway,above] {$C_k$} (HMC);

\end{tikzpicture}
\end{center}
\caption{
	\label{figure:UNC-IND-attack}
	Relation between the CPTP maps and Hilbert spaces considered in a cloning-distinguishing attack as described in \cref{definition:UNC-IND-attack}.
}
\end{figure}
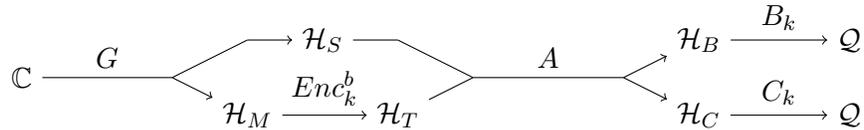
\begin{definition}[Uncloneable-Indistinguishable Security]
\label{definition:UNC-IND-security}
Let $\mc{S}$ be a QECM scheme and define $Enc_k^0$ and $Enc_k^1$ as in \cref{definition:IND-security}.

We say that $\mc{S}$ is \emph{uncloneable-indistinguishable secure} if for all cloning-distinguishing attacks $\mc{A}$ there exists a negligible function $\eta$ such that
\begin{equation}
	\E_b
	\E_{k \gets \mc{K}}
	\Tr\left[
		\left(\ketbra{b} \tensor \ketbra{b}\right)
		\left(B_k \tensor C_k\right)
		\circ
		A
		\circ
		\left(\Id_S \tensor Enc^b_k\right)
		\circ
		G(1)
	\right]
	\leq
	\frac{1}{2} + \eta(\lambda)
\end{equation}
where $\lambda$ is implicit on the left-hand side, $\mc{K}_\lambda$ is distributed on $\{0,1\}^{\kappa(\lambda)}$ such that $\Pr[\mc{K} = k] = \Tr[\ketbra{k}  K(1)]$, $B_k$ is the CPTP map defined by $\rho \mapsto B(\ketbra{k} \tensor \rho)$, and similarly for $C_k$.
\end{definition}

It is trivial to see, but worth noting, that uncloneable-indistinguishable security implies indistinguishable security. We now briefly sketch the proof.

Let $\mc{S}$ be an QECM and $\mc{A} = \left(\texttt{G}, \texttt{A}\right)$ be a distinguishing attack which shows that $\mc{S}$ is not indistinguishable secure.
Then, we can construct a cloning-distinguishing attack $\mc{A}' = \left(\texttt{G}', \texttt{A}', \texttt{B}', \texttt{C}'\right)$ which implies that $\mc{S}$ is not uncloneable-indistinguishable secure.
Set $\texttt{G}' = \texttt{G}$ and $\texttt{B}$ and $\texttt{C}$ be the circuits which do nothing on a single qubit input.
Then, we define $\texttt{A}'$ to first run $\texttt{A}$ and measure the output in the computational basis state.
The result is a single classical bit which may then be copied and given to both $\texttt{B}$ and $\texttt{C}$.
We then observe that the winning probability of $\mc{A}$ in the indistinguishable scenario is the same as the winning probability of $\mc{A}'$ in the uncloneable-indistinguishable scenario.

Finally, it can also be shown that any $0$-uncloneable secure QECM $\mc{S}$ is uncloneable-indistinguishable secure.
The proof of \cref{theorem:0-UNC=>UNC-IND} can be found in \cref{section:proofs}.

\begin{theorem}
\label{theorem:0-UNC=>UNC-IND}
Let $\mc{S}$ be an QECM.
If $\mc{S}$ is $0$-uncloneable secure and has constant message size, i.e.: $n(\lambda) = n$, then it is also uncloneable-indistinguishable secure.
\end{theorem}

%=======================%
\section{Two Protocols} %
\label{sec:protocols}   %
%=======================%

In this section, we first present a protocol for the encryption of classical messages into quantum ciphertexts based on Wiesner's conjugate encoding~(\cref{section:conjugate-encryption}). This will also include a simple proof of its uncloneable security.
Then, in \cref{section:our-protocol}, we present a refinement of this first protocol which uses quantum secure pseudorandom functions.
The proof of the uncloneable security of this protocol is more involved and so we present some technical lemmas in \cref{section:technicallemmas} before we give our final main results in \cref{sec:maintheorem}.

%------------------------------------%
\subsection{Conjugate Encryption}    %
\label{section:conjugate-encryption} %
%------------------------------------%

Our first QECM scheme is a one-time pad encoded into Wiesner states.
We emphasize that this scheme will not offer much in terms of uncloneable security but it remains an instructive example.

\begin{definition}[Conjugate Encryption]
\label{protocol:CE}
We define the \emph{conjugate encryption} QECM scheme by the following circuits, each implicitly parametrized by $\lambda$.
Note that the message size is $n(\lambda) = \lambda$, the key size is $\kappa(\lambda) = 2\lambda$ and the ciphertext size is $\ell(\lambda) = \lambda$.

\newpage % Manual layout correction.

\begin{algorithm}[H]
\DontPrintSemicolon
\caption{The key generation circuit $\texttt{Key}$.}
	\SetKwInOut{Input}{Input}
	\SetKwInOut{Output}{Output}

	\Input{None.}
	\Output{A state $\rho \in \mc{D}(\mc{Q}(\kappa))$.}

	Sample $r \gets \{0,1\}^n$ uniformly at random.\;
	Sample $\theta \gets \{0,1\}^n$ uniformly at random.\;
	Output $\rho = \ketbra{r} \tensor \ketbra{\theta}$.\;
\end{algorithm}

\begin{algorithm}[H]
\DontPrintSemicolon
\caption{The encryption circuit $\texttt{Enc}$.}
	\SetKwInOut{Input}{Input}
	\SetKwInOut{Output}{Output}

	\Input{A plaintext $m \in \{0,1\}^n$ and a key $(r, \theta) \in \{0,1\}^\kappa$.}
	\Output{A ciphertext $\rho \in \mc{D}(\mc{Q}(n))$.}

	Output $\rho = \ketbra{(m \xor r)^\theta}$.\;
\end{algorithm}

\begin{algorithm}[H]
\DontPrintSemicolon
\caption{The decryption circuit $\texttt{Dec}$.}
	\SetKwInOut{Input}{Input}
	\SetKwInOut{Output}{Output}

	\Input{A ciphertext $\rho \in \mc{D}(\mc{Q}(n))$ and a key $(r, \theta) \in \{0,1\}^{\kappa}$.}
	\Output{A plaintext $m \in \{0,1\}^n$.}
	Compute $\rho' = H^\theta \rho H^\theta$.\;
	Measure $\rho'$ in the computational basis. Call the result $c$.
	Output $c \xor r$.
\end{algorithm}
\end{definition}

The correctness of this scheme is trivial to verify and it is indistinguishable secure.
The latter follows from the fact that if $Enc_{r, \theta}^0$ and $Enc_{r, \theta}^1$ are as defined in \cref{definition:IND-security}, then for any $\rho \in \mc{D}(\mc{H}_S \tensor \mc{Q}(n))$ we have that
\begin{equation}
	\E_r
	\E_\theta
	\left(\Id_S \tensor Enc_{(r,\theta)}^1\right)(\rho)
	=
	\E_r
	\E_\theta
	\left(\Id_S \tensor Enc_{(r, \theta)}^0\right)(\rho).
\end{equation}

We will need one small technical lemma before proceeding to the proof of uncloneable security for this scheme.
The proof can be found in \cref{section:proofs}.
\begin{lemma}
\label{theorem:string-permute}
Let $n \in \N^+$, $f : \{0,1\}^n \times \{0,1\}^n \to \R$ be a function and $s \in \{0,1\}^n$ be a string.
Then,
\begin{equation}
	\E_x f(x, x \xor s) = \E_x f(x \xor s, x).
\end{equation}
\end{lemma}

\begin{theorem}
The scheme in \cref{protocol:CE} is $\lambda\log_2\left(1 + \frac{1}{\sqrt{2}}\right)$-uncloneable secure.
\end{theorem}
\begin{proof}
It suffices to show that for any cloning attack $\mc{A}$ the quantity
\begin{equation}
	\E_m
	\E_r
	\E_\theta
	\Tr\left[
		\left(\ketbra{m} \tensor \ketbra{m}\right)
		\left(B_{(r, \theta)} \tensor C_{(r, \theta)}\right)
		\circ
		A
		\left(\ketbra{(m \xor r)^\theta}\right)
	\right]
\end{equation}
is upper bounded by $\left(\frac{1}{2} + \frac{1}{2\sqrt{2}}\right)^\lambda$.
By applying \cref{theorem:string-permute} with respect to the expectation over $m$, this quantity is the same as
\begin{equation}
	\E_m
	\E_r
	\E_\theta
	\Tr\left[
		\left(\ketbra{m \xor r} \tensor \ketbra{m \xor r}\right)
		\left(B_{(r, \theta)} \tensor C_{(r, \theta)}\right)
		\circ
		A
		\left(\ketbra{m^\theta}\right)
	\right].
\end{equation}
We then see that for any fixed $r$, we can apply \cref{theorem:TFKW13-ours} to bound the expectation of the trace over $m$ and $\theta$ by $\left(\frac{1}{2} + \frac{1}{2\sqrt{2}}\right)^\lambda$.
Setting this quantity to be equal to $2^{-n+t}$, recalling that $n = \lambda$, and solving for $t$ completes the proof.
\end{proof}

Finally, note that this scheme cannot be uncloneable-indistinguishable secure if $n \geq 2$. Indeed, the adversaries could submit the all $1$ plaintext to be encrypted and split the ciphertext such that each adversary gets half of the qubits.
Once the key is revealed, the adversaries can then each obtain half of the message with probability $1$.
This is sufficient to distinguish between the two possible messages.

%----------------------------%
\subsection{Our Protocol}    %
\label{section:our-protocol} %
%----------------------------%

As discussed in \cref{sec:SummaryContributions}, the motivation for our second QECM scheme is to use quantum-secure pseudorandom functions to attempt to ``distill'' the uncloneability found in the Wiesner state.

\begin{definition}[$\mc{F}$-Conjugate Encryption]
\label{protocol:FCE}%
\label{scheme:PRF-QECM}%
\label{definition:our-protocol}%
For a function $n : \N^+ \to \N^+$, let
\begin{equation}
	\mc{F}
	=
	\left\{
		f_\lambda : \{0,1\}^\lambda \times \{0,1\}^\lambda \to \{0,1\}^{n(\lambda)}
	\right\}_{\lambda \in \N^+}
\end{equation}
be a quantum-secure pseudorandom function.
We define the $\mc{F}$-conjugate encryption QECM scheme by the following circuits which are implicitly parametrized by $\lambda$.
Note that the message size is the output size of the qPRF, $n(\lambda)$, the key size is $\kappa(\lambda) = 2\lambda$, and the ciphertext size is $\ell(\lambda) = \lambda + n(\lambda)$.

\begin{algorithm}[H]
\DontPrintSemicolon
\caption{The key generation circuit $\texttt{Key}$.}
	\SetKwInOut{Input}{Input}
	\SetKwInOut{Output}{Output}

	\Input{None.}
	\Output{A state $\rho \in \mc{D}(\mc{Q}(\kappa(\lambda)))$.}

	Sample $s \gets \{0,1\}^\lambda$ uniformly at random.\;
	Sample $\theta \gets \{0,1\}^\lambda$ uniformly at random.\;
	Output $\rho = \ketbra{s} \tensor \ketbra{\theta}$.\;
\end{algorithm}

\begin{algorithm}[H]
\DontPrintSemicolon
\caption{The encryption circuit $\texttt{Enc}$.}
	\SetKwInOut{Input}{Input}
	\SetKwInOut{Output}{Output}

	\Input{A plaintext $m \in \{0,1\}^n$ and a key $(s, \theta) \in \{0,1\}^{\kappa}$.}
	\Output{A ciphertext $\rho \in \mc{D}(\mc{Q}(\ell(\lambda)))$.}

	Sample $x \gets \{0,1\}^\lambda$ uniformly at random.\;
	Compute $c = m \xor f_\lambda(s, x)$.\;
	Output $\rho = \ketbra{c} \tensor \ketbra{x^\theta}$.\;
\end{algorithm}

\newpage % Manual layout correction.

\begin{algorithm}[H]
\DontPrintSemicolon
\caption{The decryption circuit $\texttt{Dec}$.}
	\SetKwInOut{Input}{Input}
	\SetKwInOut{Output}{Output}

	\Input{A ciphertext $\ketbra{c} \tensor \rho \in \mc{D}(\Qb(\ell))$ and a key $(s, \theta) \in \{0,1\}^{\kappa}$.}
	\Output{A plaintext $m \in \{0,1\}^n$.}

	Compute $\rho' = H^\theta \rho H^\theta$.\;
	Measure $\rho'$ in the computational basis. Call the result $r$.\;
	Output $m = c \xor f_\lambda(s, r)$.\;
\end{algorithm}
\end{definition}

It is trivial to see that this scheme is correct and we can also show that it is indistinguishable secure.
The latter follows from the fact that if we use a truly random function instead of the qPRF then, for any state $\rho$ we have
\begin{equation}
	\E_r
	\E_{H \in \Bool(\lambda, n)}
	\left(\Id_S \tensor Enc_{(r, H)}^0\right)(\rho)
	=
	\E_r
	\E_{H \in \Bool(\lambda, n)}
	\left(\Id_S \tensor Enc_{(r,H)}^1\right)(\rho)
\end{equation}
where $Enc_{(r, H)}^0$ and $Enc_{(r, H)}^1$ are as given in \cref{definition:IND-security} except that they use a truly random function $H$ instead of a keyed qPRF.
Thus, any adversary has no advantage in distinguishing the cases.
When the truly random functions are replaced by a qPRF, the adversary may have at most a negligible advantage in distinguishing these two cases.

%-------------------------------%
\subsection{Technical Lemmas}   %
\label{sec:technical lemmas}    %
\label{section:technicallemmas} %
%-------------------------------%

We first present a few technical lemmas which will be used in our proof of security.
The proof of \cref{theorem:algebra,theorem:vector-norm} appear in \cref{section:proofs}.

\begin{lemma}
\label{theorem:algebra}
Let $R$ be a ring with $a, b \in R$ and $c = a + b$.
Then, for all $n \in \N^+$, we have that
\begin{equation}
	c^n = a^n + \sum_{k = 0}^{n-1} a^{n-k-1} b c^{k}.
\end{equation}
\end{lemma}

\begin{lemma}
\label{theorem:vector-norm}
Let $\HS$ be a Hilbert space, $n \in \N^+$, and $\{v_0, v_1, \ldots, v_n\}$ be $n + 1$ vectors in $\HS$ such that $\norm{v_i} \leq 1$ for all $i \in \{1, \ldots, n\}$ and $\norm{\sum_{i = 0}^n v_i} \leq 1$.
Then,
\begin{equation}
	\norm{\sum_{i = 0}^{n} v_i}^2 \leq \norm{v_0}^2 + (3n + 2)\sum_{i = 1}^{n} \norm{v_i}.
\end{equation}
\end{lemma}

The following implicitly appears in \cite{Unr15b}.
\begin{lemma}
\label{theorem:Hxy}
Let $f : \Bool(n,m) \to \R$ be a function and $x \in \{0,1\}^n$ be a string.
For any $H \in \Bool(n,m)$ and $y \in \{0,1\}^m$, define $H_{x,y} \in \Bool(n,m)$ by
\begin{equation}
	s \mapsto
	\begin{cases}
		H(s)	& \text{ if } s \not= x,	\\
		y		& \text{ if } s = x.
	\end{cases}
\end{equation}
Then,
\begin{equation}
	\E_H f(H) = \E_H \E_y f(H_{x,y}).
\end{equation}
\end{lemma}

The following two lemmas form the core of the upcoming proofs of uncloneable security and they may be interpreted as follows.
We consider two adversaries who have oracle access to a function $H \in \Bool(\lambda, n)$  which is chosen uniformly at random.
Their goal is to simultaneously guess the value $H(x)$ for some value of $x$.
The adversaries share some quantum state which we interpret as representing all the information they may initially have on $x$.
The lemmas relate the probability of both parties simultaneously guessing $H(x)$ to their probability of being able to both simultaneously guess $x$.

The first of these lemmas, \cref{theorem:SL-NE}, considers this problem in a setting where the adversaries do not share any entanglement.
The second, \cref{theorem:SL} imposes no such restriction.

We show that the probability that both adversaries correctly guess $H(x)$ is upper bounded by
\begin{equation}
	\frac{1}{2^n} + Q \cdot G
	\hspace{1cm}
	\text{or}
	\hspace{1cm}
	9 \cdot \frac{1}{2^n} + Q' \cdot G'
\end{equation}
where $Q$ and $Q'$ are polynomial functions of the number of queries the adversaries make to the oracle and $G$ and $G'$ quantify their probability of guessing $x$ with a particular strategy.
The factor of $9$ is present only if we allow the adversaries to share entanglement.

We can interpret $G$ and $G'$ in a manner very similar to its analogous quantity in Unruh's one-way-to-hiding lemma \cite{Unr15b}.
The adversaries, instead of continuing until the end of their computation, will stop immediately before a certain (randomly chosen) query to the oracle and measure their query register in the computational basis.
Then, $G$ is related to the probability that this procedure succeeds at letting both adversaries simultaneously obtain $x$, averaged over the possible stopping points and possible functions implemented by the oracle.

The key idea in the proof of these lemmas is that we can decompose the unitary operator representing each of the adversaries' computations into two ``parts''.
Explicitly, this decomposition appears in \cref{equation:decomp1,equation:decomp2}.
One of these ``parts'' will never query the oracle on $x$ and the other could query the oracle on $x$.
We note that this idea was present in the proof of Unruh's one-way-to-hiding lemma \cite{Unr15b}. 

Recall from \cref{section:oracles} that we model queries to an oracle implementing a function $H$ as a unitary operator $O^H$ acting on a query and a response register with Hilbert spaces $\mc{H}_Q$ and $\mc{H}_R$ respectively. This unitary is defined on the computational basis states by $\ket{x}_Q \tensor \ket{y}_R \mapsto \ket{x}_Q \tensor \ket{y \xor H(x)}_R$.
A party having access to an oracle may also have some other register with Hilbert space $\mc{H}_S$ with which they perform other computations.
In general, their computation can then be modeled by an operator of the form $\left(U O^H\right)^q$ where $U$ is a unitary operator on $\mc{H}_Q \tensor \mc{H}_R \tensor \mc{H}_S$ and $q$ is the number of queries made to the oracle \cite{BBC+01,BDF+11,Unr15b}.

\begin{lemma}
\label{theorem:SL-NE}
Let $\lambda, n \in \N^+$.
For $L \in \{B, C\}$, we let
\begin{itemize}
	\item
		$s_L, q_L \in \N^+$,
	\item
		$\HS_{L_Q} = \Qb(\lambda)$, $\HS_{L_R} = \Qb(n)$, and $\HS_{L_S} = \Qb(s_L)$,
	\item
		$U_L \in \mc{U}(\HS_{L_Q} \tensor \HS_{L_R} \tensor \HS_{L_S})$, and
	\item
		$\{\pi^y_L\}_{y \in \{0,1\}^n}$ be a projective measurement on $\HS_{L_Q} \tensor \HS_{L_R} \tensor \HS_{L_S}$.
\end{itemize}
Finally, let $\ket{\psi} = \ket{\psi_B} \tensor \ket{\psi_C}$ be a separable unit vector with $\ket{\psi_L} \in \mc{Q}(n + \lambda + s_L)$ for $L \in \{B, C\}$ and $x \in \{0,1\}^\lambda$.
Then, we have
	\begin{equation}
	\E_H
	\norm{
		\Pi^{H(x)}
		\left(
			\left(U_B O^H_B\right)^{q_B}
			\tensor
			\left(U_C O^H_C\right)^{q_C}
		\right)
		\ket{\psi}
	}^2
	\leq
	\frac{1}{2^n}
	+
	(3q + 2)q
	\sqrt[4]{M}
\end{equation}
where $\Pi^{H(x)} = \pi_B^{H(x)} \tensor \pi_C^{H(x)}$, $q = q_B + q_C$ and
\begin{equation}
	M
	=
	\E_k \E_\ell \E_H \E_{H'}
	\norm{
		\left(\ketbra{x}_{B_Q} \tensor \ketbra{x}_{C_Q}\right)
		\left(
			\left(U_B O_B^H\right)^k
			\tensor
			\left(U_C O_C^{H'}\right)^\ell
		\right)
		\ket{\psi}
	}^2
\end{equation}
with $k \in \{0, \ldots, q_B-1\}$, $\ell \in \{0, \ldots, q_C-1\}$, and $H \in \Bool(\lambda,n)$.
\end{lemma}
\begin{proof}
Note that since $\ket{\psi}$ is separable, we have that
\begin{equation}
	M
	=
	\underbrace{
	\left(
		\E_H
		\E_k
		\norm{
			\ketbra{x}_{B_Q}
			\mc{O}^{U_B, H, k}
			\ket{\psi_B}
		}^2
	\right)
	}_{= M_B}
	\cdot
	\underbrace{
	\left(
		\E_{H'}
		\E_\ell
		\norm{
			\ketbra{x}_{C_Q}
			\mc{O}^{U_C, H', \ell}
			\ket{\psi_C}
		}^2
	\right)
	}_{= M_C}
	.
\end{equation}
For the remainder of the proof, we fix $L \in \{B, C\}$ such that $M_L = \min\{M_B, M_C\}$.
Note that $\sqrt{M_L} \leq \sqrt[4]{M}$.
Once again using the fact that $\ket{\psi}$ is separable, we have that
\begin{equation}
	\E_H
	\norm{
		\Pi^{H(x)}
		\left(
			\left(U_B O^H_B\right)^{q_B}
			\tensor
			\left(U_C O^H_C\right)^{q_C}
		\right)
		\ket{\psi}
	}^2
	\leq
	\E_H
	\norm{
		\pi_L^{H(x)}
		\left(U_B O^H_B\right)^{q_L}
		\ket{\psi_L}
	}^2
	.
\end{equation}
With all this, it suffices to show that
\begin{equation}
	\E_H
	\norm{
		\pi_L^{H(x)}
		\left(U_B O^H_B\right)^{q_L}
		\ket{\psi_L}
	}^2
	\leq
	\frac{1}{2^n} + (3q_L + 2)q_L\sqrt{M_L}
\end{equation}
to obtain our result.

Let $P_L = \ketbra{x}_{L_Q}$.
Using the fact that $U_L O_L^H = U_L O_L^H P_L + U_L O_L^H (\Id - P_L)$ and \cref{theorem:algebra}, we have that
\begin{equation}
\label{equation:decomp1}
\begin{aligned}
	\left(U_L O^H_L\right)^{q_L}
	=&
	\overbrace{
		\left(
			U_L
			O^H_L
			\left(\Id - P_L\right)
		\right)^{q_L}
	}^{= V_L^H}
	+
	\\
	&
	\sum_{k = 0}^{q_L - 1}
	\underbrace{
		\left(
			U_L
			O^H_L
			\left(\Id - P_L\right)
		\right)^{q_L - k - 1}
		U_L
		O^H_L
		P_L
		\left(
			U_L
			O^H_L
		\right)^k
	}_{= W_L^{H,k}}
\end{aligned}
\end{equation}
and we define $W^H_L = \sum_{k = 0}^{q_L - 1} W_L^{H, k}$ so that $\left(U_L O_L^H\right)^{q_L} = V_L^H + W_L^H$.

Using \cref{theorem:vector-norm}, the definition of the various $W$ operators, and properties of the operator norm on projectors and unitary operators, we have that
\begin{equation}
\begin{aligned}
	\norm{
		\pi^{H(x)}_L
		\left(
			V_L^H
			+
			W_L^H
		\right)
		\ket{\psi_L}
	}^2
	\leq
	\norm{
		\pi^{H(x)}_L
		V_L
		\ket{\psi_L}
	}^2
	+
	(3 q_L + 2) q_L
	\E_k
	\norm{
		P_L
		\left(U_L O_L^H\right)^k
		\ket{\psi_L}
	}
	.
\end{aligned}
\end{equation}
Using Jensen's inequality, the above inequality, and the definition of $M_L$, we have that
\begin{equation}
	\E_H
	\norm{
		\pi_L^{H(x)}
		\left(\left(U_L O_L^H\right)^{q_L}\right)
		\ket{\psi_L}
	}^2
	\leq
	\E_H
	\norm{
		\pi_L^{H(x)}
		V_L
		\ket{\psi_L}
	}^2
	+
	(3q_L + 2) q_L \sqrt{M_L}
\end{equation}
and so it suffices to show that $\E_H\norm{\pi_L^{H(x)} V_L^H \ket{\psi_L}}^2 \leq 2^{-n}$.
By \cref{theorem:Hxy}, it is then sufficient to show that
\begin{equation}
	\E_H \E_y \norm{\pi_L^y V_L^{H_{x,y}} \ket{\psi_L}}^2 \leq 2^{-n}
\end{equation}
where $H_{x, y} \in \Bool(\lambda, n)$ is defined by $H_{x,y}(x) = y$ and $H_{x,y}(s) = H(s)$ for all $s \not= x$.
Recall that $V_L^H$ is independent of the value of $H(x)$, in the sense that $V_L^{H_{x,y}} = V_L^H$ for all $y \in \{0,1\}^n$.
Indeed, prior to every query to $H$ in $V_L^H$, we project the state on a subspace which does not query $H$ on $x$.
So, using the fact that each $\pi_L^y$ projects on mutually orthogonal subspaces and that $\norm{V_L^H} \leq 1$, we have that
\begin{equation}
	\E_y\norm{\pi_L^y V_L^{H_{x,y}} \ket{\psi_L}}^2
	=
	\frac{1}{2^n}
	\norm{V_L^H \ket{\psi_L}}^2
	\leq
	\frac{1}{2^n}
\end{equation}
which completes the proof.
\end{proof}

\begin{lemma}
\label{theorem:SL}
Let $\lambda, n \in \N^+$.
For $L \in \{B, C\}$, we let
\begin{itemize}
	\item
		$s_L, q_L \in \N^+$,
	\item
		$\HS_{L_Q} = \Qb(\lambda)$, $\HS_{L_R} = \Qb(n)$, and $\HS_{L_S} = \Qb(s_L)$,
	\item
		$U_L \in \mc{U}(\HS_{L_Q} \tensor \HS_{L_R} \tensor \HS_{L_S})$, and
	\item
		$\{\pi^y_L\}_{y \in \{0,1\}^n}$ be a projective measurement on $\HS_{L_Q} \tensor \HS_{L_R} \tensor \HS_{L_S}$.
\end{itemize}
Finally, let $\ket{\psi} \in \mc{Q}(2(\lambda + n) + s_B + s_C)$ be a unit vector and $x \in \{0,1\}^\lambda$.
Then, we have
	\begin{equation}
	\E_H
	\norm{
		\Pi^{H(x)}
		\left(
			\left(U_B O^H_B\right)^{q_B}
			\tensor
			\left(U_C O^H_C\right)^{q_C}
		\right)
		\ket{\psi}
	}^2
	\leq
	\frac{9}{2^n}
	+
	(3q_B q_C + 2)q_Bq_C
	\sqrt{M}
\end{equation}
where $\Pi^{H(x)} = \pi_B^{H(x)} \tensor \pi_C^{H(x)}$ and
\begin{equation}
	M
	=
	\E_k \E_\ell \E_H
	\norm{
		\left(\ketbra{x}_{B_Q} \tensor \ketbra{x}_{C_Q}\right)
		\left(
			\left(U_B O_B^H\right)^k
			\tensor
			\left(U_C O_C^H\right)^\ell
		\right)
		\ket{\psi}
	}^2
\end{equation}
with $k \in \{0, \ldots, q_B-1\}$, $\ell \in \{0, \ldots, q_C-1\}$, and $H \in \Bool(\lambda, n)$.
\end{lemma}
\begin{proof}
For $L \in \{B, C\}$, we define $P_L = \ketbra{x}_{L_Q}$.
Using \cref{theorem:algebra} and the fact that $U_L O_L^H = U_L O_L^H P_L + U_L O_L^H (\Id - P_L)$, we have that
\begin{equation}
\label{equation:decomp2}
\begin{aligned}
	\left(U_L O^H_L\right)^{q_L}
	=&
	\overbrace{
		\left(
			U_L
			O^H_L
			\left(\Id - P_L\right)
		\right)^{q_L}
	}^{= V_L^H}
	+
	\\
	&
	\sum_{k = 0}^{q_L - 1}
	\underbrace{
		\left(
			U_L
			O^H
			\left(\Id - P_L\right)
		\right)^{q_L - k - 1}
		U_L
		O^H_L
		P_L
		\left(
			U_L
			O^H_L
		\right)^k
	}_{= W_L^{H,k}}
\end{aligned}
\end{equation}
and we define $W^H_L = \sum_{k = 0}^{q_L - 1} W_L^{H, k}$.
This implies that
\begin{equation}
\begin{aligned}
	&
	\norm{
		\Pi^{H(x)}
		\left(
			\left(U_B O_B^H \right)^{q_B}
			\tensor
			\left(U_C O_C^H \right)^{q_C}
		\right)
		\ket{\psi}
	}^2
	\\
	&=
	\norm{
		\Pi^{H(x)}
		\left(
			\left(O_B O_B^H\right)^{q_B}
			\tensor
			V_C^H
			+
			V_B^H
			\tensor
			W_C^H
			+
			W_B^H
			\tensor
			W_C^H
		\right)
		\ket{\psi}
	}^2
	.
\end{aligned}
\end{equation}
We now claim that contribution from the $W_B^H \tensor W_C^H$ operator corresponds to the $M$ in the upper bound provided in the statement.
Indeed, using \cref{theorem:vector-norm}, the definition of the various $W$ operators, and properties of the operator norm on projectors and unitary operators, we have that
\begin{equation}
\begin{aligned}
	&
	\norm{
		\Pi^{H(x)}
		\left(
			\left(O_B O_B^H\right)^{q_B}
			\tensor
			V_C^H
			+
			V_B^H
			\tensor
			W_C^H
			+
			W_B^H
			\tensor
			W_C^H
		\right)
		\ket{\psi}
	}^2
	\\
	&\leq
	\norm{
		\Pi^{H(x)}
		\left(
			\left(O_B O_B^H\right)^{q_B}
			\tensor
			V_C^H
			+
			V_B^H
			\tensor
			W_C^H
		\right)
		\ket{\psi}
	}^2
	\\
	&\phantom{\leq} +
	(3 q_B q_C + 2) q_B q_C
	\E_k
	\E_\ell
	\norm{
		\left(P_B \tensor P_C\right)
		\left(
			\left(U_B O_B^H\right)^k
			\tensor
			\left(U_C O_C^H\right)^\ell
		\right)
		\ket{\psi}
	}
	.
\end{aligned}
\end{equation}
Using Jensen's inequality, the above inequality and the definition of $M$, we have that
\begin{equation}
\begin{aligned}
	&
	\E_H
	\norm{
		\Pi^{H(x)}
		\left(
			\left(O_B O_B^H\right)^{q_B}
			\tensor
			V_C^H
			+
			V_B^H
			\tensor
			W_C^H
			+
			W_B^H
			\tensor
			W_C^H
		\right)
		\ket{\psi}
	}^2
	\\
	&\leq
	\E_H
	\norm{
		\Pi^{H(x)}
		\left(
			\left(O_B O_B^H\right)^{q_B}
			\tensor
			V_C^H
			+
			V_B^H
			\tensor
			W_C^H
		\right)
		\ket{\psi}
	}^2
	+
	(3 q_B q_C + 2) q_B q_C
	\sqrt{M}
	.
\end{aligned}
\end{equation}

It now suffices to show that
\begin{equation}
	\E_H
	\norm{
		\Pi^{H(x)}
		\left(
			\left(O_B O_B^H\right)^{q_B}
			\tensor
			V_C^H
			+
			V_B^H
			\tensor
			W_C^H
		\right)
		\ket{\psi}
	}^2
	\leq
	\frac{9}{2^n}.
\end{equation}
By \cref{theorem:Hxy}, this is equivalent to showing that
\begin{equation}
	\E_H
	\E_y
	\norm{
		\Pi^{y}
		\left(
			\left(U_B O_B^{H_{x,y}}\right)^{q_B}
			\tensor
			V_C^{H_{x,y}}
			+
			V_B^{H_{x,y}}
			\tensor
			W_C^{H_{x,y}}
		\right)
		\ket{\psi}
	}^2
	\leq
	\frac{9}{2^n}
\end{equation}
In fact, it will be sufficient to show that for any particular $H$, the expectation over $y$ is bounded by $9 \cdot 2^{-n}$.
If, for any $H$, we define
\begin{equation}
	\alpha = \E_y \norm{\Pi^y \left(\left(U_B O_B^{H_{x,y}}\right) \tensor V_C^{H_{x,y}}\right) \ket{\psi}}^2
\end{equation}
and
\begin{equation}
	\beta = \E_y \norm{\Pi^y \left(V_B^{H_{x,y}} \tensor W_C^{H_{x,y}}\right) \ket{\psi}}^2
\end{equation}
then, using the triangle inequality and the fact that the operators in $\{\Pi^y\}_{y \in \{0,1\}^n}$ project on mutually orthogonal subspaces, we have that
\begin{equation}
	\E_y
	\norm{
		\Pi^{y}
		\left(
			\left(O_B O_B^{H_{x,y}}\right)^{q_B}
			\tensor
			V_C^{H_{x,y}}
			+
			V_B^{H_{x,y}}
			\tensor
			W_C^{H_{x,y}}
		\right)
		\ket{\psi}
	}^2
	\leq
	\alpha + \beta + 2 \sqrt{\alpha \beta}.
\end{equation}
Now, noting that $V_B^{H_{x,y}}$ and $V_C^{H_{x,y}}$ do not depend on the value of $y$, as they always project on a subspace which does not query the oracle $H$ on $x$, and using properties of the operator norm, we have that
\begin{align}
	\alpha
	&=
	\E_y \norm{\Pi^y \left(\left(U_B O_B^{H_{x,y}}\right) \tensor V_C^{H_{x,y}}\right) \ket{\psi}}^2
	\\
	&\leq
	\E_y \norm{\left(\Id_B \tensor \pi^y_C\right) \left(\Id_B \tensor V_C^{H_{x,y}}\right) \ket{\psi}}^2
	\\
	&=
	\frac{1}{2^n} \norm{\left(\Id_B \tensor V_C^H\right) \ket{\psi}}^2
	\leq
	\frac{1}{2^n}.
\end{align}
A similar reasoning yields that $\beta \leq 4 \cdot 2^{-n}$, where the $4$ is a result of squaring the upper bound
\begin{equation}
\norm{W_C^{H_{x,y}}} \leq \norm{\left(U_CO^{H_{x,y}}_C\right)^{q_C}} + \norm{V_C^{H_{x,y}}} \leq 2.
\end{equation}
Finally, noting that $\alpha + \beta + 2 \sqrt{\alpha \beta} \leq 9 \cdot 2^{-n}$ finishes the proof.
\end{proof}

%-------------------------%
\subsection{Main Results} %
\label{sec:maintheorem}   %
%-------------------------%

We now have all the necessary tools to prove our main results.

\begin{theorem}
\label{thm:main-security-F-encryption}
Let $\mc{S}$ be the QECM scheme defined in \cref{scheme:PRF-QECM}.
If the qPRF is modeled by a quantum oracle, then $\mc{S}$ is $\log_2(9)$-uncloneable secure.
\end{theorem}

When we say that we model a qPRF as a quantum oracle, we mean that we model the adversaries' evaluations of the qPRF as queries to an oracle.
Specifically, if the key used in the encryption was $(s, \theta)$, we assume that the adversaries do not receive $s$ when the key is broadcasted, but rather that they receive quantum oracle access to the function $f_\lambda(s, \cdot)$.
Essentially, we are assuming that the adversaries only use $s$ to compute the function $f_\lambda(s, \cdot)$ and that they treat it as a black box.
Recall that we indicate that a circuit has oracle access to a function by placing this function in superscript to the circuit and its CPTP map.

\begin{proof}
Let $\mc{A} = (\texttt{A}, \texttt{B}, \texttt{C})$ be a cloning attack against $\mc{S}$ as described in \cref{definition:UA}.
We need to show that the probability that the adversaries can simultaneously guess a message chosen uniformly at random is upper bounded by $9 \cdot 2^{-n} + \eta(\lambda)$ for a negligible function $\eta$.
Furthermore, since the adversaries treat the qPRF as an oracle, it suffices to show that their winning probability is upper bounded by $9 \cdot 2^{-n} + \eta(\lambda)$ when averaged over all functions in $\Bool(\lambda, n)$ and not only the functions $\{f_\lambda(s, \cdot)\}_{s \in \{0,1\}^\lambda}$.
Indeed, by definition of a qPRF, their winning probability in both cases can only differ by a negligible function of $\lambda$.

The remainder of the proof is an application of \cref{theorem:SL} followed by application of \cref{theorem:TFKW13-ours}.

Accounting for the randomness of the encryption and for a fixed and implicit~$\lambda$, the quantity we wish to bound is then given by
\begin{equation}
	\omega
	=
	\E_H
	\E_\theta
	\E_x
	\E_m
	\Tr\left[
		P^m
		\left(B^H_\theta \tensor C^H_\theta\right)
		\circ
		A
		\left(
			\ketbra{m \xor H(x)} \tensor \ketbra{x^\theta}
		\right)
	\right]
\end{equation}
where $P^m = \ketbra{m} \tensor \ketbra{m}$ and $H \in \Bool(\lambda, n)$.
Then, by using \cref{theorem:string-permute} with respect to the expectation over $m$ to move the dependence on the string $H(x)$ from the state to the projector, we have that
\begin{equation}
	\omega
	=
	\E_H
	\E_\theta
	\E_x
	\E_m
	\Tr\left[
		P^{m \xor H(x)}
		\left(B^H_\theta \tensor C^H_\theta\right)
		\circ
		A
		\left(
			\ketbra{m} \tensor \ketbra{x^\theta}
		\right)
	\right]
	.
\end{equation}

Using standard purification arguments, we add auxiliary states $\ketbra{\text{aux-B}}$ and $\ketbra{\text{aux-C}}$ to the state $A(\ketbra{m} \tensor \ketbra{x^\theta})$, replace the CPTP maps $B^H_\theta$ and $C^H_\theta$ by unitary operators on the resulting larger Hilbert spaces and similarly replace the projectors $\ketbra{m}$ by projectors $\{\pi_B^m\}_{m \in \{0,1\}^n}$ and $\{\pi_C^m\}_{m \in \{0,1\}^n}$ on these larger Hilbert spaces.

Following \cite{BDF+11}, these purified unitary operators will be of the form $\left(U^\theta_L O_L^H\right)^{q_L}$, acting on a Hilbert space of the form $\mc{Q}(\lambda)_{L_Q} \tensor \mc{Q}(n)_{L_R} \tensor \mc{Q}(s_L)_{L_S}$ for some $q_L, s_L \in \N^+$ as they model oracle computations.
In particular, we note that $q_L$ represents the number of queries made to the oracle by that particular party.
We also assume that
\begin{equation}
\begin{aligned}
	\rho^{m, x, \theta}
	&=
	A(\ketbra{m} \tensor \ketbra{x^\theta}) \tensor \ketbra{\text{aux-B}} \tensor \ketbra{\text{aux-C}}
	\\
	&\in
	\mc{D}(\mc{Q}(\lambda)_{B_Q} \tensor \mc{Q}(n)_{B_R} \tensor \mc{Q}(s_B)_{B_S} \tensor \mc{Q}(\lambda)_{C_Q} \tensor \mc{Q}(n)_{C_R} \tensor \mc{Q}(s_C)_{C_S}).
\end{aligned}
\end{equation}
Finally, we can write $\rho^{m, x, \theta}$ as an ensemble of pure states, which is to say that
\begin{equation}
	\rho^{m, x, \theta} = \sum_{i \in I^{m, x, \theta}} p_i \ketbra{\psi_i^{m, x, \theta}}
\end{equation}
for some index set $I^{m, x, \theta}$, some non-zero $p_i$ which sum to $1$, and some unit vectors $\ket{\psi_i^{m, x, \theta}}$.
It then follows that $\omega$ can be expressed as
\begin{equation}
	\E_m
	\E_\theta
	\E_x
	\E_H
	\sum_{i \in I^{m, x, \theta}}
	p_i
	\norm{
		\left(\pi_B^{m\xor H(x)} \tensor \pi_C^{m\xor H(x)}\right)
		\left(
			\left(U^\theta_B O_B^H\right)^{q_B}
			\tensor
			\left(U^\theta_C O_C^H \right)^{q_C}
		\right)
		\ket{\psi_i^{m,x,\theta}}
	}^2
	.
\end{equation}
Noting that we can bring the expectation with respect to $H$ into the summation, we can then use \cref{theorem:SL} to upper bound $\omega$ by
\begin{equation}
\label{eq:use-lemma}
	\frac{9}{2^n}
	+
	q
	\E_m \E_\theta \E_x \sum_{i \in I^{m, x, \theta}} p_i
	\sqrt{
		\E_H
		\E_k
		\E_\ell
		\norm{
			Q_x
			\left(
				\left(U_BO_B^H\right)^{q_B}
				\tensor
				\left(U_C O_C^H \right)^{q_C}
			\right)
			\ket{\psi_i^{m, x, \theta}}
		}^2
	}
\end{equation}
where $q = (3 q_B q_C + 2) q_B q_C$ and $Q_x = \ketbra{x}_{Q_B} \tensor \ketbra{x}_{Q_C}$.
Defining
\begin{equation}
	\beta_{x}^{\theta, H, k}
	=
	\left(\left(U_B^\theta O_B^H\right)^{q_B}\right)^\dag
	\ketbra{x}_{Q_B}
	\left(\left(U_B^\theta O_B^H\right)^{q_B}\right),
\end{equation}
and similarly for $\gamma_x^{\theta, H, \ell}$ by replacing every instance of $B$ with $C$, we use Jensen's lemma to bring the remaining expectations and sums into the square root and obtain
\begin{equation}
	\omega
	=
	\frac{9}{2^n} + q \sqrt{
		\E_m
		\E_\theta
		\E_x
		\E_H
		\E_k
		\E_\ell
		\Tr\left[
			\left(\beta_x^{\theta, H, k} \tensor \gamma^{\theta, H, k} \right)
			\rho^{m, x, \theta}
		\right]
	}.
\end{equation}
Letting $\Phi_m$ to be the CPTP map defined by
\begin{equation}
	\rho
	\mapsto
	A \left(\ketbra{m} \tensor \rho\right)
	\tensor
	\ketbra{\text{aux-B}}
	\tensor \ketbra{\text{aux-C}}
\end{equation}
we see that, for any fixed $H$, $k$, $\ell$, and $m$, \cref{theorem:TFKW13-ours} implies that
\begin{equation}
	\E_x
	\E_\theta
	\Tr\left[
		\left(\beta_x^{\theta, H, k} \tensor \gamma^{\theta, H, k} \right)
		\rho^{m, x, \theta}
	\right]
	\leq
	\left(\frac{1}{2} + \frac{1}{2\sqrt{2}}\right)^\lambda
\end{equation}
since $\rho^{m, x, \theta} = \Phi_m\left(\ketbra{x^\theta}\right)$.
Thus,
\begin{equation}
	\omega \leq \frac{9}{2^n} + q\left(\sqrt{\frac{1}{2} + \frac{1}{2\sqrt{2}}}\right)^\lambda.
\end{equation}
Finally, since $\texttt{B}$ and $\texttt{C}$ are efficient quantum circuits, they may query the oracle a number of time which grows at most polynomially in $\lambda$.
Thus, $q \leq p(\lambda)$ for some polynomial $p$.
Noting that $\lambda \mapsto p(\lambda) \cdot \left(\sqrt{\frac{1}{2} + \frac{1}{2\sqrt{2}}}\right)^\lambda$ is a negligible function completes the proof.
\end{proof}

We can strengthen this result if the adversaries do not share any entanglement.

\begin{theorem}
\label{thm:main-security-F-encryption-NE}
Let $\mc{S}$ be the QECM scheme given in \cref{scheme:PRF-QECM}.
If the qPRF is modeled by a quantum oracle and the adversaries cannot share any entanglement, then $\mc{S}$ is $0$-uncloneable secure.
\end{theorem}
\begin{proof}[Sketch]
Follow the proof of \cref{thm:main-security-F-encryption} using the bound given by \cref{theorem:SL-NE}, instead of \cref{theorem:SL}, in the step corresponding to \cref{eq:use-lemma}.
\end{proof}

\begin{corollary}
\label{thm:main-security-F-encryption-NE-IND}
Let $\mc{S}$ be the QECM scheme given in \cref{scheme:PRF-QECM} with constant message size, i.e.: $n(\lambda) = n$.
If the qPRF is modeled by a quantum oracle and the adversaries cannot share any entanglement, then $\mc{S}$ is indistinguishable-uncloneable secure.
\end{corollary}
\begin{proof}[Sketch]
Use \cref{thm:main-security-F-encryption-NE} with \cref{theorem:0-UNC=>UNC-IND}.
\end{proof}

%%%%%%%%%%%%%%%%%%%%%%%%%%%%
\section*{Acknowledgments} %
%%%%%%%%%%%%%%%%%%%%%%%%%%%%

This material is based upon work supported by the
U.S. Air Force Office of Scientific Research under award number FA9550-17-1-0083,
Canada's NSERC,
an Ontario ERA,
and the University of Ottawa's Research Chairs program.

%%%%%%%%%%%%%%%%
% Bibliography %
%%%%%%%%%%%%%%%%

\bibliographystyle{bibtex/bst/alphaarxiv.bst}
\bibliography{bibtex/bib/full.bib,bibtex/bib/quantum.bib}

%%%%%%%%%%%
\appendix %
%%%%%%%%%%%

%==========================%
\section{Technical Proofs} %
\label{section:proofs}     %
%==========================%

\begin{proof}[\cref{theorem:0-UNC=>UNC-IND}]
Let $\mc{S}$ be a $0$-uncloneable secure QECM scheme with constant message size, i.e.: $n(\lambda) = n$.
Let $\mc{A} = (\texttt{G}, \texttt{A}, \texttt{B}, \texttt{C})$ be a cloning-distinguishing attack against $\mc{S}$.
We will construct a cloning attack $\mc{A}' = (\texttt{A}', \texttt{B}', \texttt{C}')$ and show that the winning probability of $\mc{A}$ is at most $2^{n-1}$ times the winning probability of $\mc{A}'$.
Since $\mc{S}$ is uncloneable secure, $\mc{A}'$'s winning probability is bounded by $2^{-n} + \eta(\lambda)$, which is sufficient to prove this theorem.

We now describe the circuits in the $\mc{A}'$ attack.
A schematic representation of this construction is given in \cref{figure:UNC-IND=>IND}.
\begin{itemize}
	\item[$\texttt{A}'$:]
		Run $\texttt{G}$ from $\mc{A}$ and obtain the state $G(1) \in \mc{D}(\mc{H}_S \tensor \mc{H}_M)$.
 		Measure the $M$ register in the computational basis and call the result $m'$.
		Discard the register $M$ and keep the register~$S$.
		Then, run $\texttt{A}$ from $\mc{A}$ on the state received as input and the state that was kept in the register $S$.
		In addition, give a copy of $m'$ to both $\texttt{B}'$ and $\texttt{C}'$.
	\item[$\texttt{B}'$:]
		Run $\texttt{B}$ from $\mc{A}$ on the state obtained from $\texttt{A}'$ and the encryption key.
		Measure the output in the computational basis and if the result is $0$, output $\ketbra{0^n}$.
		If the result is $1$, output the $\ketbra{m'}$ from the string which was given by $\texttt{A}'$.
	\item[$\texttt{C}'$:]
		Analogous to $\texttt{B}'$.
\end{itemize}

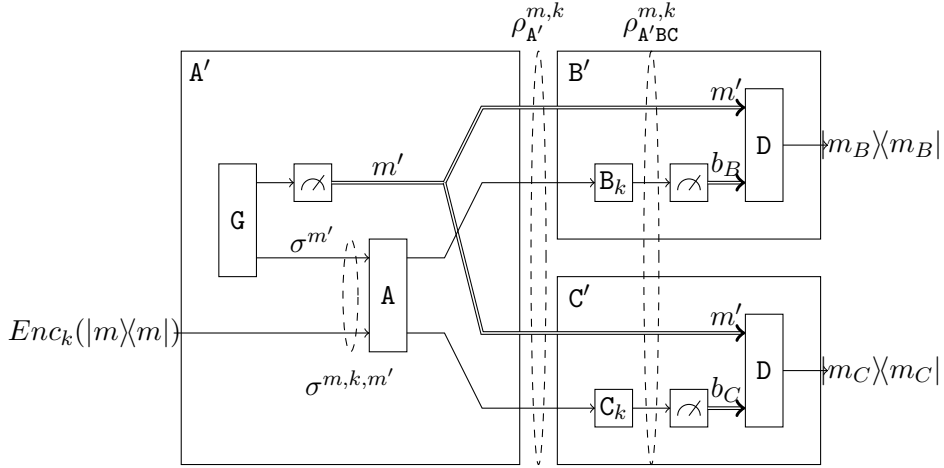
\begin{figure}
\begin{center}
\begin{tikzpicture}
	%\draw[step=0.5,gray] (0,0) grid (12,5.5);

	% The A' circuit.
	\draw (2,0) rectangle (6.5,5.5);
	\node at (2.25,5.25) {$\texttt{A}'$};
		% A circuit.
		\draw (4.5,1.5) rectangle (5,3);
		\node at (4.75,2.25) {$\texttt{A}$};
		
		% G circuit.
		\draw (2.5,2.5) rectangle (3,4);
		\node at (2.75,3.25) {$\texttt{G}$};

		% Measurement.
		\draw (3.5,3.5) rectangle (4,4);
		\draw (3.75,3.65) -- (3.9,3.85);
		\draw (3.75,3.65) ++(180:0.15) arc (180:0:.15);

		% Input.
		\node at (0.85,1.75) {$Enc_k(\ketbra{m})$};

		% G to A.
		\draw[->] (3,2.75) -- (4.5,2.75);
		% G to measurement.
		\draw[->] (3,3.75) -- (3.5,3.75);
		% Input to A.
		\draw[->] (1.9,1.75) -- (4.5,1.75);

		% Measurement to split.
		\draw[double] (4,3.75) -- (5.5,3.75);

	% The C' circuit.
	\draw (7,0) rectangle (10.5,2.5);
	\node at (7.3,2.25) {$\texttt{C}'$};
		% C circuit.
		\draw (7.5, 0.5) rectangle (8,1);
		\node at (7.75,0.75) {$\texttt{C}_k$};

		% Measurement
		\draw (8.5,0.5) rectangle (9,1);
		\draw (8.75, 0.65) -- (8.9,0.85);
		\draw (8.75, 0.65) ++(180:0.15) arc (180:0:.15);

		% D circuit.
		\draw (9.5, 0.5) rectangle (10,2);
		\node at (9.75,1.25) {$\texttt{D}$};

		% Output.
		\node at (11.3,1.25) {$\ketbra{m_C}$};	

		% C to measurement.
		\draw[->] (8,0.75) -- (8.5,0.75);
		% Measurement to D.
		\draw[->,double] (9,0.75) -- (9.5,0.75);
		% D to output.
		\draw[->] (10,1.25) -- (10.6,1.25);
	
	% The B' circuit.
	\draw (7,3) rectangle (10.5,5.5);
	\node at (7.3,5.25) {$\texttt{B}'$};
		% B circuit.
		\draw (7.5, 3.5) rectangle (8,4);
		\node at (7.75,3.75) {$\texttt{B}_k$};

		% Measurement.
		\draw (8.5,3.5) rectangle (9,4);
		\draw (8.75,3.65) -- (8.9,3.85);
		\draw (8.75,3.65) ++(180:0.15) arc (180:0:.15);

		% D circuit.
		\draw (9.5,3.5) rectangle (10,5);
		\node at (9.75,4.25) {$\texttt{D}$};

		% Output.
		\node at (11.3,4.25) {$\ketbra{m_B}$};

		% C to measurement.
		\draw[->] (8,3.75) -- (8.5,3.75);
		% Measurement to D.
		\draw[->,double] (9,3.75) -- (9.5,3.75);
		% D to output.
		\draw[->] (10,4.25) -- (10.6,4.25);

	% Interconnections
		% Split to B'D.
		\draw[->,double] (5.5,3.75) -- (6,4.75) -- (9.5,4.75);
		% Split to C'D.
		\draw[->,double] (5.5,3.75) -- (6,1.75) -- (9.5,1.75);
		% A to B.
		\draw[->] (5,2.75) -- (5.5,2.75) -- (6,3.75) -- (7.5,3.75);
		% A to C.
		\draw[->] (5,1.75) -- (5.5,1.75) -- (6, 0.75) -- (7.5,0.75);

	% Annotations.
	\node at (3.75,3) {$\sigma^{m'}$};
	\node at (4.75,4) {$m'$};
	\node at (9.25,2) {$m'$};
	\node at (9.25,5) {$m'$};
	\node at (9.25,1) {$b_C$};
	\node at (9.25,4) {$b_B$};

	\draw[dashed] (6.75,2.75) ellipse (0.1 and 2.75);
	\node at (6.75,5.9) {$\rho_{\texttt{A}'}^{m,k}$};
	\draw[dashed] (8.25,2.75) ellipse (0.1 and 2.75);
	\node at (8.25,5.9) {$\rho_{\texttt{A}'\texttt{BC}}^{m,k}$};
	\draw[dashed] (4.25,2.25) ellipse (0.1 and 0.75);
	\node at (4.25,1.1) {$\sigma^{m,k,m'}$};
	
\end{tikzpicture}
\end{center}
\caption{
	\label{figure:UNC-IND=>IND}
	A cloning-distinguishing attack $\mc{A} = (\texttt{G}, \texttt{A}, \texttt{B}, \texttt{C})$ is used to construct a cloning attack $\mc{A}' = (\texttt{A}', \texttt{B}', \texttt{C}')$.
	The $\texttt{D}$ circuit outputs $\ketbra{0^n}$ if $b = 0$ and $\ketbra{m'}$ if $b = 1$.
}
\end{figure}

We want to obtain a description of the overall state up to the point immediately after the $\texttt{B}$ and $\texttt{C}$ circuits are applied by $\texttt{B}'$ and $\texttt{C}'$.
Conditioned on the message $m$ being encrypted with the key $k$, we will denote this state by $\rho^{m,k}_{\texttt{A}'\texttt{BC}}$.
From this state, we will be able to determine the winning probability of $\mc{A}'$.

Note that $\texttt{A}'$'s first step is to obtain $G(1)$ and measure the $M$ register in the computational basis.
We define $p_{m'} = \Tr\left[(I_S \tensor \ketbra{m'}_{M}) G(1) \right]$ to be the probability that $m'$ is measured and
\begin{equation}
	\sigma^{m'}
	=
	\Tr_{M}\left[
	\frac{
		(I_S \tensor \ketbra{m'}_M) G(1) (I_S \tensor \ketbra{m'}_M)
	}{
		p_{m'}
	}
	\right]
\end{equation}
to be the post measurement state conditioned on this result and after tracing out the $M$ register.
Defining
\begin{equation}
	\sigma^{m,k,m'}
	=
	\sigma_{m'} \tensor Enc_k(\ketbra{m})
\end{equation}
allows us to write the output state of $\texttt{A}'$, conditioned on $m$ being originally encrypted with the key $k$, as
\begin{equation}
	\rho^{m, k}_{\texttt{A'}}
	=
	\sum_{m' \in \{0,1\}^n}
		p_{m'}
		A(\sigma^{m,k,m'}) \tensor \ketbra{m'} \tensor \ketbra{m'}.
\end{equation}
Thus, we have that
\begin{equation}
	\rho^{m,k}_{\texttt{A'}\texttt{B}\texttt{C}}
	=
	\sum_{m' \in \{0,1\}^n}
		p_{m'}
		(B_k \tensor C_k) \circ A (\sigma^{m, k, m'})
		\tensor
		\ketbra{m'} \tensor \ketbra{m'}.
\end{equation}
To compute $\mc{A}'$'s winning probability on message $m$ and key $k$, we define
\begin{align}
	q_{m,k}
	&=
	\Tr\left[
		(\ketbra{1} \tensor \ketbra{1} \tensor \ketbra{m} \tensor \ketbra{m})
		\rho_{\texttt{A'BC}}^{m,k}
	\right] \\
	&=
	p_m
	\Tr\left[
		(\ketbra{1} \tensor \ketbra{1})
		(B \tensor C) \circ A (\sigma^{m,k,m})
	\right]
\end{align}
and note that if $m \not= 0^n$, then $\mc{A}'$'s winning probability is given by $q_m$.
If $m = 0^n$, we must also account for the possibility that the measurements after the $\texttt{B}$ and $\texttt{C}$ circuits both output $0$. Thus, $\mc{A}'$'s winning probability in this case is at least
\begin{equation}
	q_{0^n, k}
	+
	\Tr\left[
		(\ketbra{0} \tensor \ketbra{0})
		\left(
			(B_k \tensor C_k) \circ A\left(\sum_{m' \in \{0,1\}^n} p_{m'} \sigma^{0^n,k,m'}\right)
		\right)
	\right]
\end{equation}
as we ignore any winning scenarios where the measurement results are different.

We then have that $\mc{A}'$'s winning probability is at least
\begin{equation}
\label{equation:A'-winning}
	\E_{k \gets \mc{K}}
	\frac{1}{2^n}
	\left(
	\sum_{m \in \{0,1\}^n}
	q_{m, k}
	+
	\Tr\left[
		(\ketbra{0} \tensor \ketbra{0})
		\left(
			(B_k \tensor C_k) \circ A\left(\sigma^{0^n,k}\right)
		\right)
	\right]
	\right)
\end{equation}
where $\sigma^{0^n,k} = \sum_{m' \in \{0,1\}^n} p_{m'} \sigma^{0^n, k, m'}$.
Since $\mc{S}$ is uncloneable secure, there exists a negligible function $\eta$ such that \cref{equation:A'-winning} is upper bounded by $2^{-n} + \eta(\lambda)$.
This implies that
\begin{equation}
	\label{equation:A-winning}
	\E_{k \gets \mc{K}}
	\frac{1}{2}
	\left(
	\sum_{m \in \{0,1\}^n}
	q_{m, k}
	+
	\Tr\left[
		(\ketbra{0} \tensor \ketbra{0})
		\left(
			(B_k \tensor C_k) \circ A\left(\sigma^{0^n,k}\right)
		\right)
	\right]
	\right)
\end{equation}
is upper bounded by $\frac{1}{2} + 2^{n-1}\eta(\lambda)$.
Noting that \cref{equation:A-winning} is precisely $\mc{A}$'s winning probability and that $\lambda \mapsto 2^{n-1}\eta(\lambda)$ is negligible completes the proof.
\end{proof}

\begin{proof}[\cref{theorem:TFKW13-ours}]
It suffices to apply \cref{theorem:TFKW13} with the state
\begin{equation}
	\rho
	=
	\left(\Id_A \tensor \Phi_{A'}\right) \ketbra{\text{EPR}_\lambda}_{A A'}
	=
	\frac{1}{2^\lambda}\sum_{r,s \in \{0,1\}^\lambda}
	\ketbra{r}{s} \tensor \Phi\left(\ketbra{r}{s}\right),
\end{equation}
which is the result of applying the map $\Phi$ to the second half of $\lambda$ EPR pairs.
Note that for all $\theta \in \{0,1\}^\lambda$ we have that
\begin{equation}
	\frac{1}{2^\lambda}\sum_{r,s \in \{0,1\}^\lambda}
	\ketbra{r}{s} \tensor \Phi\left(\ketbra{r}{s}\right)
	=
	\frac{1}{2^\lambda}\sum_{r,s \in \{0,1\}^\lambda}
	\ketbra{r^\theta}{s^\theta} \tensor \Phi\left(\ketbra{r^\theta}{s^\theta}\right).
\end{equation}
We then have that
\begin{equation}
\begin{aligned}
	&
	\E_\theta
	\sum_{x \in \{0,1\}^\lambda}
	\Tr\left[
		\left(\ketbra{x^\theta} \tensor B_x^\theta \tensor C_x^\theta\right)
		\rho
	\right]
	\\
	&=
	\E_\theta
	\frac{1}{2^\lambda}
	\sum_{x,r,s \in \{0,1\}^\lambda}
	\Tr\left[
		\left(\ketbra{x^\theta} \tensor B_x^\theta \tensor C_x^\theta\right)
		\left(\ketbra{r^\theta}{s^\theta} \tensor \Phi\left(\ketbra{r^\theta}{s^\theta}\right)\right)
	\right]
	\\
	&=
	\E_\theta
	\frac{1}{2^\lambda}
	\sum_{x \in \{0,1\}^\lambda}
	\Tr\left[
		\left(B_x^\theta \tensor C_x^\theta\right)
		\Phi\left(\ketbra{x^\theta}\right)
	\right].
\end{aligned}
\end{equation}
Thus the bound given in \cref{theorem:TFKW13} is directly applicable.
\end{proof}

\begin{proof}[\cref{theorem:string-permute}]
Recall that for any fixed string $s \in \{0,1\}^n$, the map $x \mapsto x \xor s$ is a permutation which is its own inverse.
If we define the map $g : \{0,1\}^n \to \R$ by $x \mapsto f(x, x\xor s)$, we then have that
\begin{equation}
	\E_x f(x, x\xor s) = \E_x g(x) = \E_x g(x \xor s) = \E_x f(x \xor s, x)
\end{equation}
which concludes the proof.
\end{proof}

\begin{proof}[\cref{theorem:algebra}]
Proceed by induction over $n$ noting that
\begin{align}
	\left(a^n + \sum_{k = 0}^{n-1} a^{n-k-1}bc^k \right) c
	&=
	a^{n + 1} + \sum_{k = 0}^{(n+1)-1} a^{(n + 1)-k-1} b c^k.
\end{align}
\end{proof}

\begin{proof}[\cref{theorem:vector-norm}]
We first note that $\norm{v_0} \leq \norm{\sum_{i = 0}^n v_i} + \sum_{i = 1}^n \norm{v_i} \leq 1 + n$.
Then, using the triangle inequality, we have that
\begin{equation}
	\norm{\sum_{i = 0}^{n} v_i}^2
	\leq
	\left(\sum_{i = 0}^{n} \norm{v_i}\right)^2
	=
	\sum_{i = 0}^n \sum_{j = 0}^n \norm{v_i} \cdot \norm{v_j}.
\end{equation}
We consider the summands in the right hand side differently depending on the value of $i$.
If $i = 0$, we note that
\begin{equation}
	\sum_{j = 0}^n \norm{v_0} \cdot \norm{v_j}
	\leq
	\norm{v_0}^2 + (n + 1)\sum_{j = 1}^{n} \norm{v_j}.
\end{equation}
If $i \not= 1$, we note that
\begin{equation}
	\sum_{j = 0}^n \norm{v_i} \cdot \norm{v_j}
	\leq
	\norm{v_i} \sum_{j = 0}^n \norm{v_j}
	\leq
	\left((n + 1) + n\right) \norm{v_i}
	.
\end{equation}
We obtain the result by adding each of these bounds.
\end{proof}

%%%%%%%%%%%%%%%%%%
%%%%%%%%%%%%%%%%%%
\end{document}